\documentclass[11pt]{amsart}
\usepackage{graphicx}
\usepackage{amscd}
\usepackage{amsmath}
\usepackage{amsfonts}
\usepackage{amssymb}
\usepackage{setspace}
\usepackage{enumerate}         
\usepackage{color}
\usepackage{url}
\usepackage{amsthm}
\usepackage{hyperref}
\usepackage{bm}
\usepackage{xy}
\usepackage{color}
\usepackage{subfig}
\usepackage{bbm}
\usepackage{etoolbox}
\allowdisplaybreaks[4]

\usepackage{natbib}

\usepackage{geometry}
\usepackage{tikzit}

\tikzstyle{red dot}=[fill=red, draw=black, shape=circle]
\tikzstyle{green dot}=[fill=green, draw=black, shape=circle]
\tikzstyle{blue dot}=[fill=blue, draw=black, shape=circle]
\tikzstyle{emptyRectanle}=[fill=white, draw=black, shape=rectangle]
\tikzstyle{emptyCircle}=[fill=white, draw=black, shape=circle]
\tikzstyle{Big Square}=[fill=white, draw=black, shape=rectangle, minimum width=2cm, minimum height=1cm]

\tikzstyle{new edge style 0}=[->]
\tikzstyle{new edge style 1}=[-]

\theoremstyle{plain}
\newtheorem{theorem}{Theorem}[section]

\newtheorem{corollary}[theorem]{Corollary}

\newtheorem{lemma}[theorem]{Lemma}
\newtheorem{notation}[theorem]{Notation}

\newtheorem{proposition}[theorem]{Proposition}

\newtheorem{definition}[theorem]{Definition}

\newtheorem{assumption}[theorem]{Assumption}

\theoremstyle{remark}
\newtheorem{remark}[theorem]{Remark}

\theoremstyle{hp}

\numberwithin{equation}{section}

\newcommand{\rsto}{]\!\kern-1.8pt ]}
\newcommand{\lsto}{[\!\kern-1.7pt [}

\numberwithin{equation}{section}

\newcommand{\FF}{\mathbb{F}}
\newcommand{\GG}{\mathbb{G}}

\newcommand{\RR}{\mathbb{R}}
\newcommand{\QQ}{\mathbb{Q}}

\newcommand{\PP}{\mathbb{P}}

\newcommand{\EE}{\mathbb{E}}

\newcommand{\cG}{\mathcal{G}}
\newcommand{\cH}{\mathcal{H}}

\newcommand{\cS}{\mathcal{S}}

\newcommand{\Ex}[2]{\mathbb{E}^{#1}\left[#2\right]}                     

\renewcommand{\cite}{\citet}

\makeatletter
\patchcmd{\@maketitle}
  {\ifx\@empty\@dedicatory}
  {\ifx\@empty\@date \else {\vskip3ex \centering\footnotesize\@date\par\vskip1ex}\fi
   \ifx\@empty\@dedicatory}
  {}{}
\patchcmd{\@adminfootnotes}
  {\ifx\@empty\@date\else \@footnotetext{\@setdate}\fi}
  {}{}{}
\newcommand{\subjclassname@JEL}{JEL Classification}
\makeatother

\begin{document}

\title[Multi-currency Valuation]{Cross Currency Valuation and Hedging in the Multiple Curve Framework}

\author{Alessandro Gnoatto}
\address[Alessandro Gnoatto]{University of Verona, Department of Economics, \newline
\indent via Cantarane 24, 37129 Verona, Italy}
\email[Alessandro Gnoatto]{alessandro.gnoatto@univr.it}

\author{Nicole Seiffert}
\address[Nicole Seiffert]{Mathematisches Institut der LMU M\"unchen, \newline
\indent Theresienstr. 39, 80333 M\"unchen, Germany}
\email[Nicole Seiffert]{nicole-seiffert@t-online.de}

\begin{abstract}
We generalize the results of \cite{BieRut15} on funding and collateralization to a multi-currency framework and link their results with those of \cite{pit12},  \cite{mopa17}, and \cite{fushita10c}.

In doing this, we provide a complete study of absence of arbitrage in a multi-currency market where, in each single monetary area, multiple interest rates coexist. We first characterize absence of arbitrage in the case without collateral.

After that we study collateralization schemes in a very general situation: the cash flows of the contingent claim and those associated to the collateral agreement can be specified in any currency. We study both segregation and rehypothecation and allow for cash and risky collateral in arbitrary currency specifications. Absence of arbitrage and pricing in the presence of collateral are discussed under all possible combinations of conventions.

Our work provides a reference for the analysis of wealth dynamics, we also provide valuation formulas that are a useful foundation for cross currency curve construction techniques. Our framework provides also a solid foundation for the construction of multi-currency simulation models for the generation of exposure profiles in the context of xVA calculations.
\end{abstract}


\date{\today}

\maketitle

\section{Introduction}
Since the 2007-2009 financial crisis several assumptions underlying financial valuation have been questioned. Several spreads have emerged (more precisely widened) between certain interest rates (notably between overnight and unsecured Ibor rates) and these rates in turn differ from interest rates agreed in the context of repurchase agreements (repo rates). From a modeling perspective this resulted in the development of multi curve interest rate models as in \cite{hen07}, \cite{bia10}, \cite{mopa10}, \cite{mer10b}, \cite{Henr14}, \cite{GPSS14}, \cite{CGNS:13} \cite{Cuchiero2016} and \cite{Cuchiero2019} among others.

Even before the financial crisis, financial institutions employed many different funding strategies to support their trading activity. Borrowing cash from the internal treasury desk (as implicitly assumed in classical asset pricing theory) is only one among different possibilities to fund a transaction. Even before the crisis repurchase agreements and collateralization constituted a possible and understood way to finance cash flows, mainly aimed at managing counterparty risk.

In a collateralization agreement, the agents participating in a transaction regularly exchange cash flows in order to reduce the outstanding exposure of a contract. A collateralized transaction is, in a nutshell, very similar in its nature to a transaction on a futures contract, where margin calls are regularly exchanged. One important difference is that in a collateralized transaction, the party who receives collateral typically pays an interest to the party who posts the collateral (either in the form of cash or shares of a risky asset with low volatility/high rating). 

The financial crisis implied a more widespread adoption of such alternative funding strategies. Collateralization agreements have now become a common aspects in the business relation between financial institutions. 

If we couple the increased importance of collateralization agreements with the emergence of spreads between interest rates, we understand that some care is needed in the context of valuation and hedging. Since interest rates differ and since multiple sources of funding are possible, then one needs to carefully model the funding policy in order to obtain pricing formulas that are consistent with the contractual conditions of a certain transaction. If multiple sources of funding are employed, the spreads between the interest rates linked to the different sources of funding must be taken into account.

The problem described above has given rise to a large stream of literature aiming at reconciling the theory behind arbitrage free valuation with the current market setting. We cite, among others, \cite{pit10}, \cite{Castagna2011}, \cite{papebri2011}, \cite{PaPeBri12}, \cite{Antonov2015}, \cite{crepey2015a}, \cite{crepey2015b}, \cite{BBPL:15} and  \cite{bripa2014ccp}. The contribution of \cite{BieRut15} consists in presenting a sound martingale pricing framework that accounts for funding costs and collateralization. Such framework is then reconciled with the results of \cite{pit10} in a pure diffusive setting. The contributions mentioned above restrict themselves to a single currency framework.

Funding strategies and collateralization agreements become even more involved as soon as we allow for multiple currencies. Financial institutions may fund the trading activity in any currency. Also collateral might be posted in arbitrary currencies. From the perspective of the hedger, i.e. the party who shorts a contingent claim, collateral might be posted or received either in domestic currency, or in the currency of the contractual cash flows, or even in a third foreign currency. Collateralization agreements might grant the collateral provider the option to choose the currency he/she uses to post collateral, thus providing the option to post collateral by using the currency with cheapest funding cost. Such feature is often referred to as collateral choice option. The presence of collateral choice options turns the valuation of even plain vanilla payoffs into a non-trivial problem.

Collateralization in multiple currencies has been analyzed already in some contributions. \cite{pit12} studies funding strategies in multiple currencies by using FX swaps as basic collateralized instrument to create funding strategies in multiple currencies. He describes the cash flows of a collateralized FX swap contract (a combination of a spot and forward FX transaction) and from such analysis he obtains the dividend process of the collateralized FX swap, which depends on the collateral rate agreed between the two counterparties of the FX swap. According to Piterbarg, such collateral rate is unrelated to the domestic or the foreign collateral rate in the two economies involved in the transaction.

\cite{fushita10b} provide a valuation formula for contingent claims with currency dislocations between contractual and collateral cashflows. Their choice of the Num\'eraire is the unsecured funding rate and the drift of the exchange rate they obtain is in line with the classical single curve theory: it is the difference between the domestic and the foreign unsecured funding rate. Concerning the contribution of \cite{fushita10b} \cite{pit12} observes that the rate of the FX swap he obtains corresponds to the difference of unsecured funding rates in \cite{fushita10b}. The approach of \cite{pit12} has been later expanded by \cite{mopa17}. Anyway, even though the underlying assumptions of \cite{pit12}, \cite{fushita10b} and \cite{mopa17} are slightly different between each other, they reach similar conclusions in terms of pricing formulas and model dynamics. \textcolor{black}{For an account of cross currency valuation the single curve pre-crisis framework, we refer the reader to \cite{MusRut}. A recent reference focusing on the cross currency valuation of CDS (quanto-CDS) in a single curve setup is given by \cite{BRIGO2019}}.

Our aim in the present paper is to generalize the martingale pricing approach of \cite{BieRut15} and show that we can generalize martingale pricing to cover the results in the references above. In line with \cite{BieRut15} we exclude the possibility of default of both agents. While this assumption can in principle be relaxed to develop a full xVA framework, it is important to stress that the multiple curve phenomenon is a distinct feature with respect to the credit quality of agents involved in a specific transaction. With our framework we will be able to capture the discrepancy between: collateral rates, repo rates and we will also allow for the presence of cross currency bases. Such spreads are publicly observed and are independent of the credit quality of the agents involved in the transaction. Our choice of assuming that agents are default-free allows us to focus on pure funding aspects. Also, such choice allows us to derive valuation formulas for market instruments that can be used in the context of curve bootstraping/calibration. 

Unlike \cite{pit12} and \cite{mopa17} we do not postulate that contracts are natively collateralized. In our view collateralization is a feature of the relation between the hedger and the counterparty and absence of arbitrage should be guaranteed irrespective of the presence or not of collateralization agreements, hence our first objective is to discuss absence of arbitrage in an uncollateralized multi currency market in the presence of multiple interest rates. In a world market with an arbitrary number of currencies $L$, in each currency area we allow for country specific submarkets with $d_{k_1}$ risky assets $k_1=1\ldots,L$ and each risky assets has a dedicated funding (repo) account. Each currency area features an unsecured funding account and we employ strategies based on such unsecured accounts to construct arbitrage free transactions on the spot foreign exchange rate. The benefit of such approach is twofold: we do not need to introduce derivatives to discuss absence of arbitrage (FX swaps involve a spot and a forward transaction) so that we can discuss absence of arbitrage of the market featuring only underlying securities, and we also disentangle the issue absence of arbitrage from the description/modelization of collateralization agreements. In this sense we are following more closely the approach of \cite{fushita10b} which we fully map to the setting of \cite{BieRut15}. This constitutes the topic of Section \ref{sec:MultiCurrTrading} and Section \ref{sec:PricingUnsecured}

In Section \ref{sec:CollateralizedTrading}, we introduce collateralization agreements. We deliberately choose to closely follow the presentation of \cite{BieRut15} and we present collateralization under both segregation, rehypothecation and we allow for collateral to be posted in the form of cash or units of a risky asset. Our extension involves the possibility that the collateral is posted/received in an arbitrary currency $k_3\in\{1,\ldots,L\}$. The findings of Section \ref{sec:CollateralizedTrading} allow us to discuss in Section \ref{sec:pricingExogColl} pricing of contingent claims in the presence of collateral under any currency. We obtain first general formulas extending \cite{BieRut15}. Later we specialize our valuation formulas in a pure diffusive setting in Section \ref{sec:DiffusionModels} which extends the literature in two ways: on the one side, we obtain pricing formulas consistent with \cite{fushita10b} and \cite{mopa17} extending \cite{BieRut15}, on the other side, based on the findings of Section \ref{sec:PricingUnsecured} we provide a sound construction of cross currency diffusion models in the presence of multiple interest rates in each single currency. Such cross currency models are of paramount importance in the context of xVA computations in the industry: as explained in e.g. \cite{Cesari2009}, \cite{green2015}, \cite{listag2015} and \cite{sokol2014}, the market standard for xVA involves the computation of valuation adjustment at the level of the full portfolio as a way to capture the beneficial effect of netting agreements and this in turn implies the need to construct high dimensional Monte Carlo simulation models simultaneously covering all risk factors in all currencies relevant for the portfolio between the hedger and the counterparty. 

In the present paper, we treat contingent claims by means of a generic process of finite variation. It is clear that such process can be used to model a single claim or a whole collection of claims that share the same funding policy. A generalization to more complex portfolio structures is feasible but beyond the scope of the present paper. For a discussion of portfolio effects in the presence of multiple legal agreements between the hedger and the counterparty we refer the reader to \cite{BiaGnoOli2019}.

\section{Multi-currency trading under funding costs}\label{sec:MultiCurrTrading}
We follow the notations of \cite{BieRut15}. We fix a finite time horizon $T>0$. Let $\left(\Omega,\cG,\GG,\PP\right)$ be a filtered probability space where the filtration $\GG=\left(\cG\right)_{t\in [0,T]}$ satisfies the usual conditions. We assume that $\cG_0$ is trivial. All processes to be introduced in the sequel are assumed to be $\GG$-adapted RCLL semimartingales.

Let $k_1$, $k_1=1,\ldots, L$, $L\in\mathbb{N}$ be an index for different currency areas. For some $k_1$, $k_1=e$, which corresponds to the domestic currency. Let $S^{i,k_1}$ denote the ex-dividend price of the $i$-th risky asset traded in unit of currency $k_1$, $i=1,\ldots, d_{k_1}$, where $d_{k_1}$ is the number of risky assets traded in terms of the currency with index $k_1$. Every asset has a cumulative dividend stream $D^{i,k_1}$. As in \cite{BieRut15} we do not postulate that the processes $S^{i,k_1}$, $i=1,\ldots, d_{k_1}$,  $k_1=1,\ldots, L$ are positive.

The trading desk can use different sources of funding, each being represented by a suitable family of cash accounts. For unsecured funding, we assume that the trading desk can fund her activity by unsecured borrowing or lending in different currencies, hence we introduce the cash accounts $B^{0,k_1}=B^{k_1}$, $k_1=1,\ldots, L$.

For every risky asset, we have an asset-specific funding account, which we call repo-account. We introduce $B^{i,k_1}$ as the funding account associated to the asset $S^{i,k_1}$. 

\textcolor{black}{In case borrowing and lending rates differ, e.g. in Section \ref{sec:pricingExogColl} on collateralization, we introduce the further superscripts $b$ and $l$ and consider, in place of a generic cash account $B^{i,k_1}$, the borrowing and lending cash accounts $B^{i,k_1,b}$ and $B^{i,k_1,l}$}.

We introduce a notation for foreign exchange rates. Let $\mathcal{X}^{k_0,k}$, $k_0,k=1,\ldots,L$ the price of one unit of currency $k$ in terms of currency $k_0$. In terms of the usual FORDOM convention in currency markets we have, e.g. for EURUSD, that $\mathcal{X}^{USD,EUR}$ is the price in USD of 1 EUR.

\begin{assumption} We introduce the following processes:
\begin{enumerate}
\item[i)] ex-dividend price process $S^{i,k_1}$, $i=1,\ldots,d_{k_1}$, $k_1=1,\ldots,L$ are real-valued RCLL semimartingales.
\item[ii)] cumulative dividend streams $D^{i,k_1}$ , $i=1,\ldots,d_{k_1}$, $k=1,\ldots,L$ are processes of finite variation with $D^{i,k_1}_0=0$.
\item[iii)] exchange rate processes $\mathcal{X}^{k_0,k}$, $k_0,k=1,\ldots,L$ are positive-valued RCLL semimartingales.
\textcolor{black}{\item[iv)] funding accounts $B^{i,k_1}$ $i=0,\ldots,d_{k_1}$ are strictly positive and continuous processes of finite variation with $B^{i,k_1}_0=1$.  Positive or negative dividends from the $i$-$k_1$-th risky asset are invested in the corresponding funding account $B^{i,k_1}$}.
\end{enumerate}
\end{assumption}

In line with \cite{BieRut15} we assume that prices are real-valued (for example, the price of an interest rate swap might be negative), foreign exchange rates are however assumed to be strictly positive. Based on the last item of the above assumption, we introduce the following objects.

\begin{definition} The cumulative dividend price $S^{i,cld,k_1}$ in units of currency $k_1$ is given as
\begin{align}
S^{i,cld,k_1}_t:=S^{i,k_1}_t+B^{i,k_1}_t\int_{(0,t]}\left(B^{i,k_1}_u\right)^{-1}dD^{i,k_1}_u.
\end{align}
The cumulative dividend $S^{i,cld,k_0,k_1}$ of the asset traded in units of currency $k_1$, expressed in units of currency $k_0$ is given as
\begin{align}
S^{i,cld,k_0,k_1}_t:=S^{i,k_1}_t\mathcal{X}^{k_0,k_{1}}_t+B^{i,k_1}_t\int_{(0,t]}\left(B^{i,k_1}_u\right)^{-1}\mathcal{X}^{k_0,k_{1}}_udD^{i,k_1}_u.
\end{align}
the discounted cumulative dividend price $\hat{S}^{i,cld,k_1}:=(B^{i,k_1})^{-1}S^{i,cld,k_1}$ in units of currency $k_1$ satisfies
\begin{align}
\hat{S}^{i,cld,k_1}_t:=\hat{S}^{i,k_1}_t+\int_{(0,t]}\left(B^{i,k_1}_u\right)^{-1}dD^{i,k_1}_u.
\end{align}
The discounted cumulative dividend $\hat{S}^{i,cld,k_0,k_1}:=(B^{i,k_1})^{-1}S^{i,cld,k_0,k_1}_t$ of the asset traded in units of currency $k_1$, expressed in units of currency $k_0$ satisfies
\begin{align}
\hat{S}^{i,cld,k_0,k_1}=\hat{S}^{i,k_1}_t\mathcal{X}^{k_0,k_{1}}_t+\int_{(0,t]}\left(B^{i,k_1}_u\right)^{-1}\mathcal{X}^{k_0,k_{1}}_udD^{i,k_1}_u
\end{align}
\end{definition}
\textcolor{black}{For the sake of clarity, let us state the following.}
\begin{notation} \textcolor{black}{Concerning currency indeces, we will make the following choices in the sequel of the paper:
\begin{itemize}
\item $k_0$ will be used for the currency of denomination of the portfolio. In most cases we will have $k_0=e$, where $e$ represents the domestic currency, so that we will omit this index in the computations.
\item $k_1$ will be used for the native currency of denomination of risky assets, associated repo cash accounts and unsecured funding accounts.
\item $k_2$ will be used for the native currency of denomination of contractual cashflows.
\item $k_3$ will be introduced in Section \ref{sec:CollateralizedTrading} as the currency of denomination of the collateral.
\end{itemize}}
\end{notation}

\subsection{Contracts and trading strategies}

\begin{definition}
A dynamic portfolio, denoted as $\varphi=(\xi, \psi)$ with 
\begin{align}
\varphi=(\xi, \psi) =\left(\xi^{1,1}, \ldots, \xi^{d_1,1,},\xi^{1,2},\ldots ,\xi^{d_L,L}, \psi^{0,1}, \ldots, \psi^{d_1,1},\psi^{0,2},\ldots, \psi^{d_L,L}\right),
\end{align}
consists of risky securities $S^{i,k_1}$, $i=1,\ldots, d_{k_1}$,  $k_1=1,\ldots, L$, the cash accounts $B^{0,k_1}=B^{k_1}$, $k_1=1,\ldots, L$, for unsecured borrowing and lending, and funding/repo-accounts $B^{i,k_1}$, $i=1,\ldots, d_{k_1}$,  $k_1=1,\ldots, L$ used for funding of the $i$-$k_1$-th risky asset.
\end{definition}
We will use the shorthand $\psi^{k_1}=\psi^{0,k_1}$, $k_1=1,\ldots, L$.  In line with \cite{BieRut15} we consider the following contracts.
\begin{definition}\label{def:bilContr}
By a bilateral financial contract, or simply a contract, we mean an arbitrary RCLL process of finite variation, denoted by $A^{k_2}$ to emphasize that the contract is denominated in terms of currency $k_2$. The process $A^{k_2}$ is aimed to represent the cumulative cash flows of a given contract from time $0$ until its maturity date $T$. By convention, we set $A^{k_2}_{0-}=0$.
\end{definition}
The process $A^{k_2}$ represents the flows from the perspective of the hedger and includes the initial flow $A^{k_2}_0$ taking place at the contract's inception. As shown in \cite{BieRut15} it can be used to describe contracts with multiple cash flows during the contract's lifetime, with the cash flow at time $0$ representing the (yet to be determined) price $p^{k_2}$ of the claim, in units of currency $k_2$. For example, in the case of a European call option written on the exchange rate $\mathcal{X}^{e,k_2}$, one has $A^e_{t}=p^e \mathbf{1}_{[0, T]}(t)-\left(\mathcal{X}^{e,k_2}_{T}-K\right)^{+} \mathbf{1}_{[T]}(t)$, $K>0$.

However, cash flows of a contract might be expressed in any currency, hence we introduce also the notation $A^{k_0,k_2}$, to denote the flows of contracts natively denominated in units of currency $k_2$, when expressed in units of currency $k_0$. Assuming $\int_{(0,t]}\mathcal{X}^{k_0,k_2}_udA^{k_2}_u$ is a square integrable random variable, for any choice of the indices $k_0,k_2$ we write
\begin{align}
A_t^{k_0,k_2}:=p^{k_2}\mathcal{X}^{k_0,k_2}_0\mathbf{1}_{[0, T]}(t)\textcolor{black}{+}\int_{(0,t]}\mathcal{X}^{k_0,k_2}_udA^{k_2}_u
\end{align}
and set $p^{k_0}:=p^{k_2}\mathcal{X}^{k_0,k_2}_0$. For example a call option written on a generic asset $S^{i,k_2}$, has the following stream of cash flows in units of domestic currency $$A^{e,k_2}_{t}=p^e \mathbf{1}_{[0, T]}(t)-\mathcal{X}^{e,k_2}_{T}\left(S^{i,k_2}_T-K\right)^{+} \mathbf{1}_{[T]}(t),$$ for $K>0$.

\begin{definition}
A trading strategy is a triplet $(x, \varphi, A^{k_2})$, where $x$ is the initial endowment of the hedger, $\varphi$ represents the hedging portfolio and $A^{k_2}$ are contractual cash flows in currency $k_2$.
\end{definition}

We denote by $V^{k_0}(x, \varphi, A^{k_2})$ the wealth process of the trading strategy $(x, \varphi, A^{k_2})$ expressed under currency $k_0$. When $k_0=e$ we simply omit the currency index and write $V(x, \varphi, A^{k_2})=V^e(x, \varphi, A^{k_2})$. We have $V_0(x, \varphi, 0)=x$ and $V_0(x, \varphi, A^{k_2})=x+A^{e,k_2}_0=x+p^e$. We introduce the following regularity assumption.

\begin{assumption} We assume that
\begin{enumerate}
\item[i)] $\xi^{i,k_1}$ $i=1,2, \ldots, d_k$, $k_1=1,\ldots,L$ are arbitrary $\GG$-predictable processes.
\item[ii)] $\psi^{j,k_1}$ $j=0,1, \ldots, d_k$, $k_1=1,\ldots,L$ are arbitrary $\GG$-adapted processes.
\end{enumerate}
all processes above are such that the stochastic integrals used in what follows are well defined.
\end{assumption}
Let us introduce the concept of self-financing trading strategy.

\begin{definition}\label{def:firstSelfFinancing} Given the hedger's initial endowment $x$, we say that a trading strategy $(x, \varphi, A^{k_2})$, associated with a contract $A^{k_2}$, $k_2\in\{1,\ldots,L\}$ is self financing, whenever the wealth process $V(x, \varphi, A^{k_2})$, which is given by the formula
\begin{align}
\label{eq:FirstValueProcess}
V_t(x, \varphi, A^{k_2})=\sum_{k_1=1}^L\mathcal{X}^{e,k_1}_t\left(\sum_{i=1}^{d_{k_1}} \xi_{t}^{i,k_1} S_{t}^{i,k_1}+\sum_{j=0}^{d_{k_1}} \psi_{t}^{j,k_1} B_{t}^{j,k_1}\right),
\end{align}
satisfies
\begin{align}
\begin{aligned}
V_t(x, \varphi, A^{k_2})&=x+\sum_{k_1=1}^L\left\{\sum_{i=1}^{d_{k_1}}\left[\int_{(0, t]} \mathcal{X}_{u}^{e, k_{1}} \xi_{u}^{i, k_{1}}\left(d S_{u}^{i, k_{1}}+d D_{u}^{i, k_{1}}\right)\right.\right.\\
&\left.+\int_{(0, t]} \xi_{u}^{i, k_{1}} S_{u}^{i, k_{1}} d \mathcal{X}_{u}^{e, k_{1}}+\int_{(0, t]} \xi_{u}^{i, k_{1}} d\left[S^{i, k_{1}}, \mathcal{X}^{e, k_{1}}\right]_u\right]\\
&\left.+\sum_{j=0}^{d_{k_1}}\left[\int_{0}^{t} \mathcal{X}_{u}^{e, k_{1}} \psi_{u}^{j, k_{1}} d B_{u}^{j, k_{1}}+\textcolor{black}{\int_{(0,t]}} \psi_{u}^{j, k_{1}} B_{u}^{j, k_{1}} d \mathcal{X}_{u}^{e, k_{1}}\right]\right\}+A^{e,k_2}_t.
\end{aligned}
\end{align}
\end{definition}

\begin{remark} In a single currency case Definition~\ref{def:firstSelfFinancing} corresponds to Definition 2.3 in \cite{BieRut15}.
\end{remark}

\subsection{Basic multi-currency setting} Absence of arbitrage is a feature of the market that must hold irrespective of the particular funding strategy adopted: Absence of arbitrage should hold irrespective of the presence or absence of a collateralization agreement. Absence of arbitrage should hold first in a basic setting without any collateralization agreement. The introduction of collateralization agreements should be done in such a way as to preserve absence of arbitrage. In this section we start our discussion of absence of arbitrage. With this aim in mind, following \cite{BieRut15}, we introduce the multi-currency basic model.

\begin{definition} We call basic multi-currency model with funding costs a market model in which trading in funding accounts and risky assets is unconstrained.
\end{definition}

This simple setting is instrumental in analyzing more realistic models with further trading covenants. Following \cite{BieRut15},
we introduce the concept of \textit{netted wealth}, which will be instrumental in characterizing absence of arbitrage in the multi-currency market: In fact, the concept of martingale measure will be that of a measure such that the discounted netted wealth is a (local) martingale. As explained in \cite{BieRut15}, this is necessary because the wealth process includes $A^{k_2}$, and one needs to compensate such position with holdings on $-A^{k_2}$.

\begin{definition}
The netted wealth $V^{net}(x,\varphi,A^{k_2})$ of a trading strategy is defined by the equality
\begin{align}
V^{net}(x,\varphi,A^{k_2}):=V(x,\varphi,A^{k_2})+V(0,\widetilde{\varphi},-A^{k_2}),
\end{align}
where $(0,\widetilde{\varphi},-A^{k_2})$ is the unique self-financing trading strategy that uses holdings in $B^e$ to finance a position $A^{k_2}$: the trader borrows money from treasury (i.e. borrows units of $B^e$), purchases units of the currency $k_2$ (i.e. buys units of $B^{0,k_2}$) and uses them to enter a position in the claim with dividend process $A^{k_2}$, and leaves the position unhedged, meaning that for $\widetilde{\varphi}$ we have $\xi^{i,k_1}=\psi^{j,k_1}=0$ for any $i=1,\ldots,d_k$ and  $j=1,\ldots,d_k$.
\end{definition}

Notice that the net effect in $\widetilde{\varphi}$ is that of a short position in the domestic unsecured account $B^e=B^{0,e}$ and a long position on $A^{k_2}$ with two opposite positions in $B^{0,k_2}$ compensating each other. \textcolor{black}{In the following result, we use the simplified notation $B^e:=B^{0,e}$}.

\begin{lemma}\label{lem:VnetV}
The following equality holds, for all $t\in[0,T]$.
\begin{align}
\label{eq:firstLemma}
V^{net}_t(x,\widetilde{\varphi},A^{k_2})=V_t(x,{\varphi},A^{k_2})-B^e_t\int_{[0,t]}\frac{dA_u^{e,k_2}}{B^e_u}.
\end{align}
\end{lemma}

\begin{proof}
This corresponds to Lemma 2.1 in \cite{BieRut15}. We provide the details in what follows. \textcolor{black}{Thanks to \eqref{eq:FirstValueProcess},} we have $V_t(0,\widetilde{\varphi},-A^{k_2})=\tilde{\psi}^e_t B^e_t$. We also have
\begin{align*}
d\left(\frac{V(0,\widetilde{\varphi},-A^{k_2})}{B^e}\right)_t&=\frac{dV_t(0,\widetilde{\varphi},-A^{k_2})}{B^e_t}-\frac{V_t(0,\widetilde{\varphi},-A^{k_2})}{\left(B^e_t\right)^2}dB^e_t\\
&=\frac{\tilde{\psi}^e_tdB^e_t-dA^{e,k_2}_t}{B^e_t}-\frac{V_t(0,\widetilde{\varphi},-A^{k_2})}{\left(B^e_t\right)^2}dB^e_t\\
&=-\left(B^e_t\right)^2dA^{e,k_2}_t,
\end{align*}
where we used $\tilde{\psi}^e_t=\frac{V_t(0,\widetilde{\varphi},-A^{k_2})}{B^e_t}$. \textcolor{black}{In the derivation above, we used the fact that the contractual stream $A^{k_2}$, when converted in units of the domestic currency, gives the cashflow stream $A^{e,k_2}$.}  Now, since we know that $V_0(0,\widetilde{\varphi},-A^{k_2})=-A^{e,k_2}_0$, we can integrate both sides to conclude that
\begin{align*}
V_t(0,\widetilde{\varphi},-A^{k_2})=B^e_t\int_{[0,t]}\frac{dA^{e,k_2}_u}{B^e_u},
\end{align*}
and the conclusion immediately follows from the definition of netted wealth.
\end{proof}

\subsubsection{Preliminary computation in the basic model}

Following \cite{BieRut15} we introduce, for $i=1,\ldots d_{k_1}$, $k_1=1,\ldots,L$, the processes
\begin{align}
\label{eq:definitionKi}
K^{i,k_0,k_1}_t:=\int_{(0,t]}B^{i,k_1}_ud\hat{S}^{i,cld,k_0,k_1}_u.
\end{align}
This process represents the wealth, denominated in units of currency $k_0$, discounted by the funding account $B^{i,k_1}$, of a self-financing trading strategy that invests in the asset $S^{i,k_1}$. For the sake of simplicity, the next process is only considered in terms of units of the domestic currency $e$:
\begin{align}
K^{\varphi,k_2}_t:=\int_{(0,t]}B^{e}_ud\tilde{V}_u(x,\varphi,A^{k_2})-(A^{e,k_2}_t-A^{e,k_2}_0)=\int_{(0,t]}B^{e}_ud\tilde{V}^{net}_u(x,\varphi,A^{k_2}),
\end{align}
where $\tilde{V}^{net}(x,\varphi,A^{k_2}):=(B^e)^{-1}V^{net}(x,\varphi,A^{k_2})$ and $\tilde{V}(x,\varphi,A^{k_2}):=(B^e)^{-1}V(x,\varphi,A^{k_2})$ and the last equality follows from \eqref{eq:firstLemma}.

The following proposition is instrumental for the analysis of absence of arbitrage in the basic model and more advanced settings. We remark again that we are adopting the point of view of the domestic currency $e$, but analogous computations make it possible to obtain the same claims with respect to any currency denomination.

\begin{proposition}\label{prop:PrelimRes}
For any self-financing strategy $\varphi$ we have that, for every $t\in[0,T]$
\begin{align}
\label{eq:Kvarphi}
\begin{aligned}
K^{\varphi,k_2}_t&=\sum_{k_1=1}^L\sum_{i=1}^{d_{k_1}}\int_{(0,t]}\xi^{i,k_1}_udK^{i,e,k_1}_u\\
&+\sum_{k_1=1}^L\sum_{i=1}^{d_{k_1}}\int_0^t\frac{B^e_u}{B^{i,k_1}_u}\left(\psi^{i,k_1}_uB^{i,k_1}_u+\xi^{i,k_1}_uS^{i,k_1}_u\right)\mathcal{X}^{e,k_1}_ud\left(\frac{B^{i,k_1}}{B^e}\right)_u\\
&+\sum_{k_1=1}^L\sum_{i=1}^{d_{k_1}} \textcolor{black}{\int_{(0,t]}} B^{i,k_1}_u\psi^{i,k_1}_ud\mathcal{X}^{e,k_1}_u+\sum_{k_1=1}^L\int_0^tB^e_u\psi^{k_1}_u
d\left(\frac{\mathcal{X}^{e,k_1}B^{k_1}}{B^e}\right)_u.
\end{aligned}
\end{align}
Assume also that the repo constraint holds, i.e. for all $i=1,\ldots,d_{k_1}$, $k_1=1,\ldots,L$ we have
\begin{align}
\label{eq:repoConstraint}
\zeta_t^{i,k_1}:=\psi^{i,k_1}_tB^{i,k_1}_t+\xi^{i,k_1}_tS^{i,k_1}_t=0,\quad t\in[0,T],
\end{align}
then we have
\begin{align}
\label{eq:assumeRepoConstraint}
\begin{aligned}
K^{\varphi,k_2}_t&=\sum_{k_1=1}^L\sum_{i=1}^{d_{k_1}}\int_{(0,t]}\xi^{i,k_1}_uB^{i,k_1}_u\left(\mathcal{X}^{e,k_1}_ud\left(\frac{S^{i,k_1}}{B^{i,k_1}}\right)_u+\frac{\mathcal{X}^{e,k_1}_u}{B^{i,k_1}_u}dD^{i,k_1}_u+d\left[\frac{S^{i,k_1}}{B^{i,k_1}},\mathcal{X}^{e,k_1}\right]_u\right)\\
&+\sum_{k_1=1}^L\int_0^tB^e_u\psi^{k_1}_u
d\left(\frac{\mathcal{X}^{e,k_1}B^{k_1}}{B^e}\right)_u.
\end{aligned}
\end{align}
\end{proposition}

\begin{proof}
Let $V := V(x,\varphi,A^{k_2}) $ and hence $\tilde{V}:=(B^e)^{-1}V$. Then we have
\begin{align*}
d\tilde{V}_t&=-\frac{V_t}{(B_t^e)^2}dB^e_t+\frac{1}{B^e_t}\left(\sum_{k_1=1}^L\left\{\sum_{i=1}^{d_{k_1}}\left[ \mathcal{X}_{t}^{e, k_{1}} \xi_{t}^{i, k_{1}}\left(d S_{t}^{i, k_{1}}+d D_{t}^{i, k_{1}}\right)+ \xi_{t}^{i, k_{1}} S_{t}^{i, k_{1}} d \mathcal{X}_{t}^{e, k_{1}}\right.\right.\right.\\
&\left.\left.\left.+ \xi_{t}^{i, k_{1}} d\left[S^{i, k_{1}}, \mathcal{X}^{e, k_{1}}\right]_t\right]+\sum_{j=0}^{d_{k_1}}\left[ \mathcal{X}_{t}^{e, k_{1}} \psi_{t}^{j, k_{1}} d B_{t}^{j, k_{1}}+\psi_{t}^{j, k_{1}} B_{t}^{j, k_{1}} d \mathcal{X}_{t}^{e, k_{1}}\right]\right\}+dA^{e,k_2}_t\right)\\
&=-\frac{1}{(B_t^e)^2}\left(\sum_{k_1=1}^L\mathcal{X}^{e,k_1}_t\left(\sum_{i=1}^{d_{k_1}} \xi_{t}^{i,k_1} S_{t}^{i,k_1}+\sum_{j=0}^{d_{k_1}} \psi_{t}^{j,k_1} B_{t}^{j,k_1}\right)\right)dB^e_t\\
&+\frac{1}{B^e_t}\left(\sum_{k_1=1}^L\left\{\sum_{i=1}^{d_{k_1}}\left[ \mathcal{X}_{t}^{e, k_{1}} \xi_{t}^{i, k_{1}}\left(d S_{t}^{i, k_{1}}+d D_{t}^{i, k_{1}}\right)+ \xi_{t}^{i, k_{1}} S_{t}^{i, k_{1}} d \mathcal{X}_{t}^{e, k_{1}}\right.\right.\right.\\
&\left.\left.\left.+ \xi_{t}^{i, k_{1}} d\left[S^{i, k_{1}}, \mathcal{X}^{e, k_{1}}\right]_t\right]+\sum_{j=0}^{d_{k_1}}\left[ \mathcal{X}_{t}^{e, k_{1}} \psi_{t}^{j, k_{1}} d B_{t}^{j, k_{1}}+\psi_{t}^{j, k_{1}} B_{t}^{j, k_{1}} d \mathcal{X}_{t}^{e, k_{1}}\right]\right\}+dA^{e,k_2}_t\right).
\end{align*}
By regrouping terms we obtain
\begin{align*}
d\tilde{V}_t&=\sum_{k_1=1}^L\sum_{i=1}^{d_{k_1}}\xi^{i,k_1}_td\left(\frac{\mathcal{X}^{e,k_1}S^{i,k_1}}{B_t^e}\right)_t+\sum_{k_1=1}^L\sum_{i=1}^{d_{k_1}}\xi^{i,k_1}_t\frac{\mathcal{X}^{e,k_1}_t}{B^{e}_t}dD^{i,k_1}_t\\
&+\sum_{k_1=1}^L\sum_{i=1}^{d_{k_1}}\psi^{i,k_1}_td\left(\frac{\mathcal{X}^{e,k_1}B^{i,k_1}}{B^e}\right)_t+\sum_{k_1=1}^L\psi^{k_1}_td\left(\frac{\mathcal{X}^{e,k_1}B^{k_1}}{B^e}\right)_t+(B_t^e)^{-1}dA^{e,k_2}_t.
\end{align*}
We set
\begin{align}
\tilde{S}^{i,cld,e,k_1}_t:=\frac{\mathcal{X}^{e,k_1}_tS^{i,k_1}_t}{B^e_t}+\int_{(0,t]}\frac{\mathcal{X}^{e,k_1}_u}{B^{e}_u}dD^{i,k_1}_u.
\end{align}
We can then focus on $K^{\varphi,k_2}$. We have
\begin{align*}
dK_t^{\varphi,k_2}&=B^e_td\tilde{V}_t(x,\varphi,A^{k_2})-dA^{e,k_2}_t\\
&=\sum_{k_1=1}^L\sum_{i=1}^{d_{k_1}}B^e_t\xi^{i,k_1}_td\tilde{S}^{i,cld,e,k_1}_t+\sum_{k_1=1}^L\sum_{i=1}^{d_{k_1}}B^e_t\psi^{i,k_1}_td\left(\frac{\mathcal{X}^{e,k_1}B^{i,k_1}}{B^e}\right)_t\\
&\quad+\sum_{k_1=1}^LB^e_t\psi^{k_1}_td\left(\frac{\mathcal{X}^{e,k_1}B^{k_1}}{B^e}\right)_t\\
&=\sum_{k_1=1}^L\sum_{i=1}^{d_{k_1}}B^e_t\xi^{i,k_1}_td\left(\frac{S^{i,k_1}\mathcal{X}^{e,k_1}}{B^{i,k_1}}\frac{B^{i,k_1}}{B^e}\right)_t+\sum_{k_1=1}^L\sum_{i=1}^{d_{k_1}}\xi^{i,k_1}_t\mathcal{X}^{e,k_1}_tdD^{i,k_1}_t\\
&\quad +\sum_{k_1=1}^L\sum_{i=1}^{d_{k_1}}B^e_t\psi^{i,k_1}_td\left(\frac{\mathcal{X}^{e,k_1}B^{i,k_1}}{B^e}\right)_t+\sum_{k_1=1}^LB^e_t\psi^{k_1}_td\left(\frac{\mathcal{X}^{e,k_1}B^{k_1}}{B^e}\right)_t\\
&=\sum_{k_1=1}^L\sum_{i=1}^{d_{k_1}}B^e_t\xi^{i,k_1}_t\frac{S^{i,k_1}_t\mathcal{X}^{e,k_1}_t}{B^{i,k_1}_t}d\left(\frac{B^{i,k_1}}{B^e}\right)_t+\sum_{k_1=1}^L\sum_{i=1}^{d_{k_1}}\xi^{i,k_1}_tB^{i,k_1}_td\left(\frac{S^{i,k_1}\mathcal{X}^{e,k_1}}{B^{i,k_1}}\right)_t\\
&\quad +\sum_{k_1=1}^L\sum_{i=1}^{d_{k_1}}B^{i,k_1}_t\xi^{i,k_1}_t\frac{\mathcal{X}^{e,k_1}_t}{B^{i,k_1}_t}dD^{i,k_1}_t+\sum_{k_1=1}^L\sum_{i=1}^{d_{k_1}}B^e_t\psi^{i,k_1}_td\left(\frac{\mathcal{X}^{e,k_1}B^{i,k_1}}{B^e}\right)_t\\
&\quad+\sum_{k_1=1}^LB^e_t\psi^{k_1}_td\left(\frac{\mathcal{X}^{e,k_1}B^{k_1}}{B^e}\right)_t.
\end{align*}
Since $dK^{i,e,k_1}_t=B_t^{i,k_1}d\hat{S}^{i,cld,e,k_1}_t$, we have
\begin{align*}
dK_t^{\varphi,k_2}&=\sum_{k_1=1}^L\sum_{i=1}^{d_{k_1}}\xi^{i,k_1}_tdK^{i,e,k_1}_t\\
&+\sum_{k_1=1}^L\sum_{i=1}^{d_{k_1}}\left(\frac{B^e_t}{B^{i,k_1}_t}\right)\left(B^{i,k_1}_t\psi_t^{i,k_1}\mathcal{X}^{e,k_1}_td\left(\frac{B^{i,k_1}}{B^e}\right)_t+B^{i,k_1}_t\psi^{i,k_1}_t\frac{B^{i,k_1}_t}{B^e_t}d\mathcal{X}^{e,k_1}_t\right.\\
&\quad\quad\quad\quad\quad\quad\quad\quad\quad\quad\left.+\xi^{e,k_1}_tS^{i,k_1}_t\mathcal{X}^{e,k_1}_td\left(\frac{B^{i,k_1}}{B^e}\right)_t\right)+\sum_{k_1=1}^LB^e_t\psi^{k_1}_td\left(\frac{\mathcal{X}^{e,k_1}B^{k_1}}{B^e}\right)_t
\end{align*}
and \eqref{eq:Kvarphi} part is proven. Now, under the repo constraint \eqref{eq:repoConstraint}, we have
\begin{align*}
\textcolor{black}{\int_{(0,t]}} B^{i,k_1}_u\psi^{i,k_1}_u\frac{B^{i,k_1}_u}{B^e_u}d\mathcal{X}^{e,k_1}_u=-\textcolor{black}{\int_{(0,t]}} S^{i,k_1}_u\xi^{i,k_1}_u\frac{B^{i,k_1}_u}{B^e_u}d\mathcal{X}^{e,k_1}_u
\end{align*}
and so we obtain \eqref{eq:assumeRepoConstraint}.
\end{proof}

\subsubsection{Wealth dynamics in the basic model}
In \cite{BieRut15}, the single currency analogue of the increment $dK^{i,e,k_1}_t$, i.e. $dK^{i,e,e}_t$  represents the change in price of the $i$-th asset, net of funding costs. This becomes clearer thanks to the following result, that follows from an application of the Ito product rule.

\begin{lemma}The following equality holds true for $t\in[0,T]$.
\begin{align}
\label{eq:gainProcessAsset}
\begin{aligned}
K^{i,e,k_1}_t&=\int_{(0,t]}\left(S^{i,k_1}_ud\mathcal{X}^{e,k_1}_u-\frac{S^{i,k_1}_u\mathcal{X}^{e,k_1}_u}{B^{i,k_1}_u}dB^{i,k_1}_u+\mathcal{X}^{e,k_1}_udS^{i,k_1}_u\right.\\
&\left.\quad\quad\quad\quad\quad+d\left[\mathcal{X}^{e,k_1},S^{i,k_1}\right]_u+\mathcal{X}^{e,k_1}_udD^{i,k_1}_u\right).
\end{aligned}
\end{align}
\end{lemma}
Let us specialize \eqref{eq:gainProcessAsset} in a single currency setting, under the additional assumption that the funding account $B^{i,e}$ is absolutely continuous with respect to the Lebesgue measure, so that we write $dB^{i,e}_t=r^{i,e}_tB^{i,e}_tdt$, for some $\GG$-progressively measurable process $r^{i,e}=(r^{i,e}_t)_{t\in[0,T]}$. Since we obviously have $\mathcal{X}^{e,e}_t\equiv 1,\  t\in[0,T]$ we obtain
\begin{align}
\begin{aligned}
K^{i,e,e}_t&=\int_{(0,t]}\left(dS^{i,e}_u-S^{i,e}_ur^{i,e}_udu+dD^{i,k_1}_u\right).
\end{aligned}
\end{align}
Such expression is often referred to in the literature as the gain process from the $i$-th risky asset. In a multi currency setting, however, this no longer holds, since we have a further term involving the currency risk related to the foreign repo cash account: In Proposition \ref{prop:PrelimRes} we also have the term
\begin{align*}
\sum_{k_1=1}^L\sum_{i=1}^{d_{k_1}} \textcolor{black}{\int_{(0,t]}} B^{i,k_1}_u\psi^{i,k_1}_ud\mathcal{X}^{e,k_1}_u,
\end{align*}
which captures the impact of fluctuations of foreign exchange rates on the funding costs related to foreign repo positions. The definition of the martingale property for the gain process of risky assets should account also for this last source of funding costs, which is identically zero in the single-currency case.

\begin{corollary}
Formula \eqref{eq:Kvarphi} in Proposition \ref{prop:PrelimRes} is equivalent to the following expressions.
\begin{align}
\label{eq:CorollaryFirst}
\begin{aligned}
d\tilde{V}^{net}_t(x,\varphi,A^{k_2})&=\sum_{k_1=1}^L\sum_{i=1}^{d_{k_1}}\xi^{i,k_1}_t\frac{B^{i,k_1}_t}{B^e_t}d\hat{S}^{i,cld,e,k_1}_t\\
&+\sum_{k_1=1}^L\sum_{i=1}^{d_{k_1}}\frac{1}{B^{i,k_1}}\left(\psi^{i,k_1}_t B^{i,k_1}_t+\xi^{i,k_1}_t S^{i,k_1}_t\right)\mathcal{X}^{e,k_1}_td\left(\frac{B^{i,k_1}}{B^e}\right)_t\\
&+\sum_{k_1=1}^L\sum_{i=1}^{d_{k_1}}\frac{B^{i,k_1}_t}{B^e_t}\psi^{i,k_1}_td\mathcal{X}^{e,k_1}_t+\sum_{k_1=1}^L\psi^{k_1}_td\left(\frac{\mathcal{X}^{e,k_1}B^{k_1}}{B^e}\right)_t,
\end{aligned}
\end{align}
\begin{align}
\label{eq:CorollarySecond}
\begin{aligned}
d\tilde{V}_t(x,\varphi,A^{k_2})&=\sum_{k_1=1}^L\sum_{i=1}^{d_{k_1}}\xi^{i,k_1}_t\frac{B^{i,k_1}_t}{B^e_t}d\hat{S}^{i,cld,e,k_1}_t\\
&+\sum_{k_1=1}^L\sum_{i=1}^{d_{k_1}}\frac{1}{B^{i,k_1}}\left(\psi^{i,k_1}_t B^{i,k_1}_t+\xi^{i,k_1}_t S^{i,k_1}_t\right)\mathcal{X}^{e,k_1}_td\left(\frac{B^{i,k_1}}{B^e}\right)_t\\
&+\sum_{k_1=1}^L\sum_{i=1}^{d_{k_1}}\frac{B^{i,k_1}_t}{B^e_t}\psi^{i,k_1}_td\mathcal{X}^{e,k_1}_t+\sum_{k_1=1}^L\psi^{k_1}_td\left(\frac{\mathcal{X}^{e,k_1}B^{k_1}}{B^e}\right)_t+(B^e_t)^{-1}dA^{e,k_2}_t,
\end{aligned}
\end{align}
\begin{align}
\label{eq:CorollaryThird}
\begin{aligned}
d{V}_t(x,\varphi,A^{k_2})&=\tilde{V}_t(x,\varphi,A^{k_2})dB^e_t+\sum_{k_1=1}^L\sum_{i=1}^{d_{k_1}}\xi^{i,k_1}_tdK^{i,e,k_1}_t\\
&+\sum_{k_1=1}^L\sum_{i=1}^{d_{k_1}}\frac{B^e_t}{B^{i,k_1}}\left(\psi^{i,k_1}_t B^{i,k_1}_t+\xi^{i,k_1}_t S^{i,k_1}_t\right)\mathcal{X}^{e,k_1}_td\left(\frac{B^{i,k_1}}{B^e}\right)_t\\
&+\sum_{k_1=1}^L\sum_{i=1}^{d_{k_1}}B^{i,k_1}_t\psi^{i,k_1}_td\mathcal{X}^{e,k_1}_t+\sum_{k_1=1}^L B^e_t\psi^{k_1}_td\left(\frac{\mathcal{X}^{e,k_1}B^{k_1}}{B^e}\right)_t+dA^{e,k_2}_t.
\end{aligned}
\end{align}
\end{corollary}

\begin{proof}
We deduce \eqref{eq:CorollaryFirst} from $dK^{\varphi,e,k_2}_t=B^e_td\tilde{V}^{net}_t(x,\varphi,A^{k_2})$ and \eqref{eq:definitionKi}. From \eqref{eq:CorollaryFirst} we deduce \eqref{eq:CorollarySecond} due to
\begin{align*}
d\tilde{V}_t(x,\varphi,A^{k_2})=d\tilde{V}^{net}_t(x,\varphi,A^{k_2})+(B^e_t)^{-1}dA^{e,k_2}_t.
\end{align*}
Finally, using
\begin{align*}
d{V}_t(x,\varphi,A^{k_2})=d\tilde{V}_t(x,\varphi,A^{k_2})dB^e_t+B^e_td\tilde{V}_t(x,\varphi,A^{k_2}),
\end{align*}
we deduce \eqref{eq:CorollaryThird}.
\end{proof}

\begin{remark}
\textcolor{black}{In the present paper we limit ourselves to consider a linear setting without bid-ask spreads in the repo and the funding accounts that corresponds to the \textit{basic model with funding costs} of \cite{BieRut15}. In this particular case the following equality holds
\begin{align*}
V^{net}(x, \varphi, A^{k_2})=V(x, \varphi, 0)
\end{align*}
which is clear by inspecting \eqref{eq:CorollaryFirst} and \eqref{eq:CorollarySecond}. This means that, in the basic setting without bid-offer spreads in the repo and funding accounts, the introduction of $V^{net}$ is despensable and one can work instead with the process $V(x, \varphi, 0)$ in place of $V^{net}(x, \varphi, A^{k_2})$. To underline this fact while keeping at the same time a clear mapping with the notation of \cite{BieRut15}, we shall use in the sequel the more compact notation $V^{net}(x, \varphi)$.}
\end{remark}

\section{Pricing and hedging in an unsecured multi-currency market}\label{sec:PricingUnsecured}
In this section, we discuss the problem of pricing and hedging in a multi-currency market with funding costs (i.e. multiple curves for different assets) but in the absence of collateralization. This provides a sound foundation for a martingale pricing approach that we extend in subsequent sections to include collateral in different currencies. As usual, our discussion is based on and generalizes the work of \cite{BieRut15} in a single currency setting. For the sake of simplicity, we work in the setting of the basic multi-currency model with funding costs, i.e. we exclude the possibility of bid-offer spreads in the rates.

\subsection{Arbitrage for hedger}
\label{arbitrage1}
Let $x$ be the initial endowment in units of the local currency $e$. We denote by $V^0(x)$ the wealth process of a self-financing strategy $(x,\varphi^0,0)$, where $\varphi^0$ is the portfolio with all components set to zero except $\psi^{0,e}=\psi^e$. The wealth process of the strategy is simply $V^0_t(x)=xB^e_t$. Given a contract $A^{k_2}$, an arbitrage opportunity is present in the market if the hedger can create a higher netted wealth (i.e. by going long and short the contract, while leaving one position unhedged) at time $T$ than the future value of the initial endowment. We restrict ourselves to admissible trading strategies as defined by the following.

\begin{definition}\label{def:Admissible}
A self-financing trading strategy $(x,\varphi,A^{k_2})$ is admissible for the hedger whenever the discounted netted wealth $V^{net}(x,\varphi)$ is bounded from below by a constant.
\end{definition}

\begin{definition}\label{def:AoA}
An admissible trading strategy $(x,\varphi,A^{k_2})$ is an arbitrage opportunity for the hedger whenever the following conditions are satisfied
\begin{enumerate}
\item[i)] $\PP\left(V^{net}_T(x,\varphi)\geq V^0_T(x)\right)=1$,
\item[ii)] $\PP\left(V^{net}_T(x,\varphi)> V^0_T(x)\right)>0$.
\end{enumerate}
\end{definition}

\begin{remark} From $V^0_t(x)=xB^e_t$ we can rewrite the conditions in Definition~\ref{def:AoA} as 
\begin{enumerate}
\item[i)] $\PP\left(V^{net}_T(x,\varphi)\geq xB^e_T\right)=1$ $\Longleftrightarrow$ $\PP\left(\tilde{V}^{net}_T(x,\varphi)\geq x\right)=1$,
\item[ii)] $\PP\left(V^{net}_T(x,\varphi)>xB^e_t\right)>0$ $\Longleftrightarrow$ $\PP\left(\tilde{V}^{net}_T(x,\varphi)>x\right)>0$.
\end{enumerate} 
Also, from \eqref{eq:CorollaryFirst} we deduce that the condition is independent of the initial endowment. Independence of the initial endowment fails as soon as we postulate different borrowing and lending rates for unsecured positions, so that one has $V^0_t(x)=x^+B^{0,e,l}_t-x^-B^{0,e,b}_t$.
\end{remark}
A classical textbook arbitrage strategy can be constructed in a market with two locally risk-free assets growing at different rates. To preclude such trivial arbitrage opportunities, the repo constraint \eqref{eq:repoConstraint} becomes crucial. The financial meaning of the repo constraint is that the holdings on every risky asset are financed by a position on an asset-specific cash account and it is not possible to create long-short positions on different cash accounts that produce riskless profits. Intuitively speaking, this means that, for every fixed currency $k_1$, the corresponding market consists of a combination of $d_{k_1}$ sub-markets, each consisting of a single risky asset with an associated financing account.

\begin{proposition}\label{prop:aoabasemkt}
Assume that all strategies available to the hedger are admissible in the sense of Definition~\ref{def:Admissible} and satisfy the repo constraint \eqref{eq:repoConstraint}. If there exists a probability measure $\QQ^e$ on $(\Omega,\mathcal{G}_T)$, such that $\QQ^e\sim\PP$ and such that the processes
\begin{align}
\label{eq:firstMartingale}
\begin{aligned}
\left(\int_{(0,t]}\left(\mathcal{X}^{e,k_1}_ud\left(\frac{S^{i,k_1}}{B^{i,k_1}}\right)_u+\frac{\mathcal{X}^{e,k_1}_u}{B^{i,k_1}_u}dD^{i,k_1}_u+d\left[\frac{S^{i,k_1}}{B^{i,k_1}},\mathcal{X}^{e,k_1}\right]_u\right)\right)_{0\leq t\leq T},
\end{aligned}
\end{align}
\begin{align}
\label{eq:secondMartingale}
\left(\frac{\mathcal{X}^{e,k_1}_tB^{k_1}_t}{B^e_t}\right)_{0\leq t\leq T}.
\end{align}
$i=1,\ldots,d_{k_1}$, $k_1=1\ldots,L$ are local martingales, then the basic multi-currency model with funding costs is arbitrage free for the hedger.
\end{proposition}
\begin{proof}
Assume that the repo constraint \eqref{eq:repoConstraint} holds. This implies that we can write $\psi^{i,k_1}_t=-\frac{\xi^{i,k_1}_tS^{i,k_1}_t}{B^{i,k_1}_t}$. Then looking at 
\eqref{eq:CorollaryFirst} we have
\begin{align*}
\sum_{k_1=1}^L\sum_{i=1}^{d_{k_1}}\int_{(0,t]}\frac{B^{i,k_1}_u}{B^e_u}\psi^{i,k_1}_ud\mathcal{X}^{e,k_1}_u=-\sum_{k_1=1}^L\sum_{i=1}^{d_{k_1}}\int_{(0,t]}\frac{B^{i,k_1}_u}{B^e_u}\frac{\xi^{i,k_1}_uS^{i,k_1}_u}{B^{i,k_1}_u}d\mathcal{X}^{e,k_1}_u,
\end{align*}
and also
\begin{align*}
&\sum_{k_1=1}^L\sum_{i=1}^{d_{k_1}}\int_{(0,t]}\xi^{i,k_1}_u\frac{B^{i,k_1}_u}{B^e_u}d\hat{S}^{i,cld,e,k_1}_u\\
&=
\sum_{k_1=1}^L\sum_{i=1}^{d_{k_1}}\int_{(0,t]}\xi^{i,k_1}_u\frac{B^{i,k_1}_u}{B^e_u}\left(d\left(\frac{S^{i,k_1}\mathcal{X}^{e,k_1}}{B^{i,k_1}}\right)_u+\frac{\mathcal{X}^{e,k_1}_u}{B^{i,k_1}_u}dD^{i,k_1}_u\right).
\end{align*}
Using the above expressions we can write
\begin{align}
\begin{aligned}
\label{eq:vNetLocalMartingale}
\tilde{V}^{net}_t(x,\varphi)&=x+\sum_{k_1=1}^L\sum_{i=1}^{d_{k_1}}\int_{(0,t]}\xi^{i,k_1}_u\frac{B^{i,k_1}_u}{B^e_u}\left(\mathcal{X}^{e,k_1}_ud\left(\frac{S^{i,k_1}}{B^{i,k_1}}\right)_u\right..\\
&\quad\quad\quad\quad\quad\quad\quad\quad\quad\quad\left.+\frac{\mathcal{X}^{e,k_1}_u}{B^{i,k_1}_u}dD^{i,k_1}_u+d\left[\frac{S^{i,k_1}}{B^{i,k_1}},\mathcal{X}^{e,k_1}\right]_u\right)\\
&\quad\quad+\sum_{k_1=1}^L\int_0^t\psi^{k_1}_ud\left(\frac{\mathcal{X}^{e,k_1}B^{k_1}}{B^e}\right)_u.
\end{aligned}
\end{align}
Using assumptions \eqref{eq:firstMartingale} and \eqref{eq:secondMartingale}, we observe that \eqref{eq:vNetLocalMartingale} is a local martingale bounded from below by a constant, hence by Fatou Lemma it is a supermartingale.
\end{proof}

\subsection{Fair valuation in the presence of funding costs}\label{sec:pricingNoColl}
\label{Fair valuation}
We work under the assumption that the model is arbitrage free for the hedger for any contract. We wish to describe the fair price of a contract at time zero from the perspective of the hedger, (i.e. from the perspective of the seller of the contract). Recall the notation $p^e\in\RR$ for the price of the claim. Recall also that $p^e=A^{e,k_2}_0$. We use the following standard convention: if $p^e>0$ it means that the hedger receives the amount $p^e$ from the counterparty, whereas $p^e<0$ means that the hedger is paying $p^e$ to the counterparty. The following is along the lines of \cite{BieRut15}, Definition 3.5.

\begin{definition}
We say that $\bar{p}^e=A^{e,k_2}_0\in\RR$ is a hedger's fair price if, for any self-financing trading strategy $\left(x,\varphi,A^{k_2}\right)$ such that the discounted wealth $\tilde{V}\left(x,\varphi,A^{k_2}\right)$ is bounded from below by a constant we have either
\begin{align*}
\PP\left(V_T\left(x,\varphi,A^{k_2}\right)=V^0_T(x)\right)=1,
\end{align*}
or
\begin{align*}
\PP\left(V_T\left(x,\varphi,A^{k_2}\right)<V^0_T(x)\right)>0.
\end{align*}
\end{definition}
If the price $\bar{p}^e$ is too high, then we have an arbitrage as defined in the following, which is the analogue of \cite{BieRut15} Definition 3.6.

\begin{definition}
We say that a quadruplet $\left(p^e,x,\varphi,A^{k_2}\right)$, where $p^e\in\RR$ and $\left(x,\varphi,A^{k_2}\right)$ is an admissible trading strategy such that the discounted wealth process $\tilde{V}\left(x,\varphi,A^{k_2}\right)$ is bounded from below by a constant, is a hedger's arbitrage opportunity for $A^{k_2}$ at price $p^e$, if
\begin{align*}
\PP\left(V_T\left(x,\varphi,A^{k_2}\right)\geq V^0_T(x)\right)=1,
\end{align*}
or
\begin{align*}
\PP\left(V_T\left(x,\varphi,A^{k_2}\right)>V^0_T(x)\right)>0.
\end{align*}
\end{definition}
The following result characterizes the hedger's fair price and generalizes Proposition 3.2 in \cite{BieRut15}.

\begin{proposition}
Under the assumptions of Proposition~\ref{prop:aoabasemkt}, $\bar{p}^e\in\RR$ is a hedger's fair price, whenever, for any admissible trading strategy $\left(x,\varphi,A^{k_2}\right)$ satisfying the repo constraint \eqref{eq:repoConstraint}, we have either
\begin{align*}
&\PP\left(\bar{p}^e+\sum_{k_1=1}^L\sum_{i=1}^{d_{k_1}}\int_{(0,T]}\frac{B^{i,k_1}_u}{B^e_u}\left(\xi^{i,k_1}_ud\hat{S}^{i,cld,e,k_1}_u+\psi^{i,k_1}_ud\mathcal{X}^{e,k_1}_u\right)\right.\\
&\quad\quad\quad\quad\left.+\sum_{k_1=1}^L\int_{(0,T]}\psi^{k_1}_ud\left(\frac{\mathcal{X}^{e,k_1}B^{k_1}}{B^e}\right)_u+\int_{(0.T]}(B^e_u)^{-1}dA^{e,k_2}_u=0\right)=1,
\end{align*}
or
\begin{align*}
&\PP\left(\bar{p}^e+\sum_{k_1=1}^L\sum_{i=1}^{d_{k_1}}\int_{(0,T]}\frac{B^{i,k_1}_u}{B^e_u}\left(\xi^{i,k_1}_ud\hat{S}^{i,cld,e,k_1}_u+\psi^{i,k_1}_ud\mathcal{X}^{e,k_1}_u\right)\right.\\
&\quad\quad\quad\quad\left.+\sum_{k_1=1}^L\int_{(0,T]}\psi^{k_1}_ud\left(\frac{\mathcal{X}^{e,k_1}B^{k_1}}{B^e}\right)_u+\int_{(0.T]}(B^e_u)^{-1}dA^{e,k_2}_u<0\right)>0.
\end{align*}
\end{proposition}

\begin{proof} We recall Lemma~\ref{lem:VnetV} and make use of \eqref{eq:CorollaryFirst}. We have
\begin{align*}
1&=\PP\left(V_T\left(x,\varphi,A^{k_2}\right)=V^0_T(x)\right)\\
&=\PP\left(\tilde{V}_T\left(x,\varphi,A^{k_2}\right)=\frac{V^0_T(x)}{B^e_T}\right)\\
&=\PP\left(p^e+x+\sum_{k_1=1}^L\sum_{i=1}^{d_{k_1}}\int_{(0,T]}\frac{B^{i,k_1}_u}{B^e_u}\left(\xi^{i,k_1}_ud\hat{S}^{i,cld,e,k_1}_u+\psi^{i,k_1}_ud\mathcal{X}^{e,k_1}_u\right)\right.\\
&\quad\quad\quad\quad\left.+\sum_{k_1=1}^L\int_{(0,T]}\psi^{k_1}_ud\left(\frac{\mathcal{X}^{e,k_1}B^{k_1}}{B^e}\right)_u+\int_{(0.T]}(B^e_u)^{-1}dA^{e,k_2}_u=x\right),
\end{align*}
from which we obtain the first relation. The second is proven analogously.
\end{proof}

\section{Multi Currency trading under funding costs and collateralization}\label{sec:CollateralizedTrading}

We consider the situation where the hedger posts or receives collateral in a currency $k_3\in\{1,\ldots,L\}$, represented by a process $C^{k_3}$, which is right continuous and $\GG$-adapted. In the literature on counterparty credit risk the symbol $C$ is often used to denote the so-called \textit{variation margin}. Nowadays financial institutions also exchange another collateral called \textit{initial margin}, meant to provide a form of over-collateralization. Our discussion in the present section aims at covering most collateral conventions, so that the formulas we derive can be suitably combined in order to describe either variation margin or initial margin or even a situation where multiple types of collateral are present. As any random variable, $C_t$ can be split in its positive and negative part. In particular, we adopt the following convention:
\begin{itemize}
\item $C^{k_3,+}_t$ is the value of collateral received by the hedger to the counterparty in currency $k_3$.
\item $C^{k_3,-}_t$ is the value of collateral posted by the hedger from the counterparty in currency $k_3$.
\end{itemize}
We will allow for collateral to be posted in any currency and either in units of cash (which constitutes the most common form) or risky assets. For this, following \cite{BieRut15} we introduce two dedicated assets, denominated in the currency $k_3$. This means that for one of the currencies in the set $\{1,\ldots,L\}$ we will have $d_{k_3}+2$ assets. For simplicity we assume that, when collateral is posted in terms of a risky asset, the currency of the posted and received collateral are the same. The situation where the two currencies differ is a rather uncommon situation which can however be accommodated in principle in our set-up. We make the following assumptions:
\begin{itemize}
\item The risky asset $S^{d_{k_3}+1}$ is delivered by the hedger as collateral.
\item The risky asset $S^{d_{k_3}+2}$ is received by the hedger as collateral.
\item We also assume that the collateral account satisfies $C^{k_3}_T=0$, meaning that, when trading is over, the collateral is returned to its legal owner.
\item Depending on the underlying collateral convention, the hedger receives or pays interest contingent on being the poster or receiver of collateral: the hedger receives interest payments based on $B^{c,k_3,l}$ or pays interests based on $B^{c,k_3,b}$. More precisely, the amount of interest is determined by 
\begin{align}
\label{borrowlendingcond}
\eta^{k_3,l}:=(B^{c,k_3,l})^{-1} (C^{k_3})^{-} \text{ and respectively }\eta^{k_3,b}:=-(B^{c,k_3,b})^{-1} (C^{k_3})^{+}.
\end{align}

\end{itemize}

\subsection{Collateral conventions} To be self-contained, let us recall the standard conventions for collateral.
\begin{itemize}
\item \textbf{Segregation} Under segregation, the collateral amount must be kept in a separate account and is not available as a source of funding for the trading activity. The hedger, when he/she receives collateral, can not use it for trading: he/she is only allowed to receive possibly zero interest based on $B^{d_{k_3}+2,k_3}$. \color{black}The dynamics of the wealth of the hedger do not depend on the amount of cash or shares of the asset $S^{d_{k_3}+2,k_3}$ he/she receives. On the contrary the amount of cash or shares of the asset $S^{d_{k_3}+1}$ he/she posts to the counterparty has an effect on the dynamics of the portfolio. Segregation is the standard for the exchange of initial margin.
\item \textbf{Rehypothecation} Under rehypothecation, the hedger is allowed to use the cash or the shares of securities he/she receives to fund his/her trading activity. Rehypothecation constitutes the most adopted convention for variation margin.
\end{itemize}\ \\
The dynamics of the hedger's wealth differ under segregation or rehypothecation. There is also an impact of whether collateral is posted in form of cash or risky assets We also need to address the action that is undertaken by the hedger when he/she receives collateral: we need to discuss how such amount of wealth is reinvested. The reinvestment activity will be modelled by means of an additional cash account, that we wish to make specific on the fact that we have rehypothecation or segregation: to this aim we introduce the cash accounts $B^{d_{k_3}+2,k_3,s}$ in case of segregation and $B^{d_{k_3}+2,k_3,h}$ in case of rehypothecation. The following notation helps to distinguish between those cases. In a single currency case this reflects Definition 4.2 and Definition 4.3 of \cite{BieRut15}.

\begin{definition}[Cash collateral]\label{def:cashCollateral} Cash collateral is specified as follows:
\begin{itemize}
\item[(i)] If the hedger receives cash collateral denoted by $(C^{k_3})^+_t$, he/she pays interest determined by the borrowing account $B^{c,k_3,b}_t$ and $(C^{k_3})^+_t$. In case of segregation, he/she receives interest based on the amount $(C^{k_3})^+_t$ and the cash account $B^{d_{k_3}+2,k_3,s}_t$. Under rehypothecation, the hedger may use the amount $(C^{k_3})^+_t$ to fund the trading activity before maturity: in particular he/she uses units of $B^{d_{k_3}+2,k_3,h}_t$ for his/her own trading purposes. Recall that $\{B^{d_{k_3}+2,k_3,s}, B^{d_{k_3}+2,k_3,h}\}$ is a specification of $B^{d_{k_3}+2,k_3}$, monitoring the underlying collateral convention, as well for the corresponding strategy \{$\eta^{d_{k_3}+2,k_3,s},\eta^{d_{k_3}+2,k_3,h}\}$.
\item[(ii)] If the hedger posts cash collateral, he/she delivers the amount $(C^{k_3})^-_t$, borrowed from the cash account $B^{d_{k_3}+3,k_3}_t$ and receives interest determined by $B^{c,k_3,l}_t$ in return. As collateral is posted in form of cash, the following equalities hold for any $t \in [0,T]$:
\begin{align}
\xi_t^{d_{k_3}+1,k_3} =\psi_t^{d_{k_3}+1,k_3}=0, &\quad& \eta_t^{d_{k_3}+3,k_3}B_t^{d_{k_3}+3,k_3} = -(C_t^{k_3})^-.
\end{align}
\end{itemize}
\end{definition}

We assume the hedger receives or delivers shares of the risky asset $S^{d_{k_3}+1,k_3}$, which are supposed to have low credit risk and should be uncorrelated with the underlying trading portfolio. We stress the following fact: due to the assumption that the asset is uncorrelated with the underlying portfolio, in the case where the collateral is received there is no reason to include the holdings of such asset in the portfolio. Instead, if the asset is posted, the hedger needs to fund and create a position in such asset in order to fulfil the margin call.

\begin{definition}[Risky asset collateral]\label{def:riskyCollateral} Risky collateral is specified as follows:
\begin{itemize}
\item[(i)] If the hedger receives $\xi^{d_{k_3}+2,k_3}_t$ shares of the risky asset $S^{d_{k_3}+2,k_3}_t$, used as collateral, he/she is committed to pay interest to the counterparty determined by $B^{c,k_3,b}_t$ and $(C^{k_3})^+_t=\xi^{d_{k_3}+2,k_3}_tS^{d_{k_3}+2,k_3}_t$. In case of segregation where reinvesting collateral is not allowed, the hedger receives interest on a basic bank deposit determined by $B^{d_{k_3}+2,k_3,s}_t$ similar to the cash collateral case. 
\item[(ii)] If the hedger delivers a number of shares $\xi^{d_{k_3}+1,k_3}_t$ of the risky asset $S^{d_{k_3}+1,k_3}_t$ to the counterparty, funded by the account $B^{d_{k_3}+1,k_3}_t$, he/she receives interest determined by the collateral account $B^{c,k_3,l}_t$ in return. Hence, the following setting can be defined for $t \in [0,T]$:
\begin{align}
\begin{aligned}
(C_t^{k_3})^- &= \xi_t^{d_{k_3}+1,k_3}S_t^{d_{k_3}+1,k_3}, \quad \eta_t^{d_{k_3}+3,k_3}=0,\\
 & \xi_t^{d_{k_3}+1,k_3}S_t^{d_{k_3}+1,k_3} + \psi_t^{d_{k_3}+1,k_3}B_t^{d_{k_3}+1,k_3} = 0,
\end{aligned}
\end{align}
which implies $\psi_t^{d_{k_3}+1,k_3}B_t^{d_{k_3}+1,k_3} = -(C_t^{k_3})^-$ for all $t \in [0,T]$.
\end{itemize}
\end{definition}

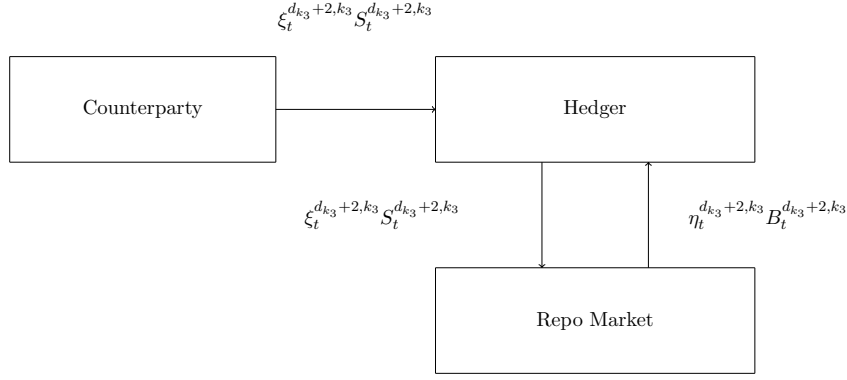
\begin{figure}
\centering
\scalebox{0.7}{
\begin{tikzpicture}
	\begin{pgfonlayer}{nodelayer}
		\node [style=none] (0) at (-2, 1) {};
		\node [style=none] (1) at (4, 1) {};
		\node [style=none] (2) at (-2, -1) {};
		\node [style=none] (3) at (4, -1) {};
		\node [style=none] (4) at (-2, -3) {};
		\node [style=none] (5) at (4, -3) {};
		\node [style=none] (6) at (-2, -5) {};
		\node [style=none] (7) at (4, -5) {};
		\node [style=none] (8) at (-10, 1) {};
		\node [style=none] (9) at (-5, 1) {};
		\node [style=none] (10) at (-10, -1) {};
		\node [style=none] (11) at (-5, -1) {};
		\node [style=none] (12) at (-5, 0) {};
		\node [style=none] (13) at (-2, 0) {};
		\node [style=none] (14) at (0, -1) {};
		\node [style=none] (15) at (0, -3) {};
		\node [style=none] (16) at (2, -3) {};
		\node [style=none] (17) at (2, -1) {};
		\node [style=none] (21) at (1, -4) {Repo Market};
		\node [style=none] (22) at (-7.5, 0) {Counterparty};
		\node [style=none] (23) at (1, 0) {Hedger};
		\node [style=none] (24) at (-3.5, 0.5) {};
		\node [style=none] (26) at (4.25, -2) {$\eta^{d_{k_3}+2,k_3}_tB^{d_{k_3}+2,k_3}_t$};
		\node [style=none] (27) at (-3.5, 1.75) {$\xi^{d_{k_3}+2,k_3}_tS^{d_{k_3}+2,k_3}_t$};
		\node [style=none] (29) at (-3, -2) {$\xi^{d_{k_3}+2,k_3}_tS^{d_{k_3}+2,k_3}_t$};
	\end{pgfonlayer}
	\begin{pgfonlayer}{edgelayer}
		\draw (0.center) to (2.center);
		\draw (2.center) to (3.center);
		\draw (3.center) to (1.center);
		\draw (1.center) to (0.center);
		\draw (4.center) to (6.center);
		\draw (6.center) to (7.center);
		\draw (7.center) to (5.center);
		\draw (5.center) to (4.center);
		\draw (8.center) to (10.center);
		\draw (10.center) to (11.center);
		\draw (11.center) to (9.center);
		\draw (9.center) to (8.center);
		\draw [style=new edge style 0] (12.center) to (13.center);
		\draw [style=new edge style 0] (14.center) to (15.center);
		\draw [style=new edge style 0] (16.center) to (17.center);
	\end{pgfonlayer}
\end{tikzpicture}
}
\caption{Collateral reinvestment process in case a risky asset is used.}
\end{figure}

At this point, it is important to stress an important aspect. During the trading activity the hedger will in general simultaneously be managing assets/amounts of cash he/she legally owns, together with assets/amounts of cash that belong to the counterparty, hence it is convenient to distinguish between the following
\begin{itemize}
\item $V_t(x,\varphi,A^{k_2},C^{k_3})$: this is the wealth of the hedger, representing the value of the portfolio of assets that belong to the hedger.
\item $V_t^p(x,\varphi,A^{k_2},C^{k_3})$: this is the value of the (full) portfolio of the hedger, including the assets/amounts of cash that belong to the counterparty.
\item $V_t^c(x,\varphi,A^{k_2},C^{k_3}):=V_t(x,\varphi,A^{k_2},C^{k_3})-V_t^p(x,\varphi,A^{k_2},C^{k_3})$ i.e. the difference between the legal wealth of the hedger and his/her portfolio, is  called adjustment process and represents the impact of collateralization.
\end{itemize}
Let us recall that the wealth processes are expressed in units of the local currency $e$. We now proceed to formally define the processes introduced above. It is rather clear that, in the absence of collateralization, we recover our previous formulation for the dynamics of the wealth process.

\begin{definition}
For any $t \in [0,T]$ and $\{k_2,k_3\} \in \{1,\ldots,L\}$, we call the process
\begin{align}
\begin{aligned}
\label{eq:vcollwealth}
V_t(x,\varphi,A^{k_2},C^{k_3}) :&= 
\sum_{k_1=1}^L \mathcal{X}_t^{e,k_1} \left(\sum_{i=1}^{d_{k_1}}\xi_t^{i,k_1}S_t^{i,k_1} + \sum_{j=0}^{d_{k_1}} \psi_t^{j,k_1}B_t^{j,k_1}\right)\\
&+\mathcal{X}_t^{e,k_3}\left(\xi_t^{d_{k_3}+1,k_3}S_t^{d_{k_3}+1,k_3} + \psi_t^{d_{k_3}+1,k_3}B_t^{d_{k_3}+1,k_3}\right.\\
&\left.+ \eta_t^{k_3,b}B_t^{c,k_3,b} + \eta_t^{k_3,l}B_t^{c,k_3,l} + \eta_t^{d_{k_3}+2,k_3}B_t^{d_{k_3}+2,k_3}+\eta_t^{d_{k_3}+3,k_3}B_t^{d_{k_3}+3,k_3}\right),
\end{aligned}
\end{align}
the extended wealth process under funding costs and collateralization, where $(x, \varphi,A^{k_2},C^{k_3})$ denotes the hedger's collateralized trading strategy for the portfolio
\begin{align}
\varphi = (\xi, \psi, \phi(k_3)) =\left(\xi^{1,1}, \ldots, \xi^{d_1,1,},\xi^{1,2},\ldots ,\xi^{d_L,L}, \psi^{0,1}, \ldots, \psi^{d_1,1},\psi^{0,2},\ldots, \psi^{d_L,L}, \phi(k_3)\right),
\end{align}
with
\begin{align*}
\phi(k_3):= (\xi^{d_{k_3}+1,k_3},\psi^{d_{k_3}+1,k_3}, \eta^{k_3,b}, \eta^{k_3,l}, \eta^{d_{k_3}+2,k_3},\eta^{d_{k_3}+3,k_3}).
\end{align*}
\end{definition}
We also introduce the following.



\begin{definition}
Let $(x,\varphi,A^{k_2}, C^{k_3})$ be a collateralized trading strategy of the hedger and $t \in [0,T]$.
\begin{itemize}
\item[(i)] The value of the hedger's portfolio $V^p(x,\varphi,A^{k_2},C^{k_3})$ at time $t$ is defined by 
\begin{align}
\begin{aligned}
V_t^p(x,\varphi, A^{k_2},C^{k_3}) :&= \sum_{k_1=1}^L \mathcal{X}_t^{e,k_1} \left(\sum_{i=1}^{d_{k_1}}\xi_t^{i,k_1}S_t^{i,k_1} + \sum_{j=0}^{d_{k_1}} \psi_t^{j,k_1}B_t^{j,k_1}\right)\\
&+ \mathcal{X}_t^{e,k_3}\left(\xi_t^{d_{k_3}+1,k_3}S_t^{d_{k_3}+1,k_3} + \psi_t^{d_{k_3}+1,k_3}B_t^{d_{k_3}+1,k_3}\right.\\
&\left.+\eta_t^{d_{k_3}+3,k_3}B_t^{d_{k_3}+3,k_3}\right).
\end{aligned}
\end{align}
\item[(ii)] In addition, denote $V^c(x,\varphi,A^{k_2},C^{k_3})$ given by
\begin{align}
\begin{aligned}
V_t^c(x,\varphi,A^{k_2},C^{k_3}):&= V_t(x,\varphi,A^{k_2},C^{k_3}) - V_t^p(x,\varphi,A^{k_2},C^{k_3})\\
&=\mathcal{X}_t^{e,k_3}\left(\eta_t^{k_3,b}B_t^{c,k_3,b} + \eta_t^{k_3,l}B_t^{c,k_3,l} + \eta_t^{d_{k_3}+2,k_3}B_t^{d_{k_3}+2,k_3}\right),
\end{aligned}
\end{align}
as adjustment process of the hedger's wealth.
\end{itemize}
\end{definition}

The adjustment process reflects the presence of a collateralization agreement between the hedger and the counterparty. Let us recall that $ \eta_t^{d_{k_3}+2,k_3}$ might be either $ \eta_t^{d_{k_3}+2,k_3,s}$ or $ \eta_t^{d_{k_3}+2,k_3,h}$, depending on the particular convention agreed by the hedger and the counterparty.

\begin{remark}
By using assumption \eqref{borrowlendingcond}, we receive for the adjustment process
\begin{align}
V_t^c(x,\varphi,A^{k_2},C^{k_3})= \mathcal{X}_t^{e,k_3} \left(-C_t^{k_3} + \eta_t^{d_{k_3}+2,k_3}B_t^{d_{k_3}+2,k_3}\right).
\end{align}
for any $t \in [0,T]$.
\end{remark}

We introduce the following useful notation.

\begin{definition}
In the collateralized multi-currency model, the process of all interest generated by the collateral account, denoted by $F^{h}$ under rehypothecation and $F^{s}$ under segregation, is given by
\begin{align}
\label{def:Foverallinterest}
F_t^h := F_t^c + \int_0^t \mathcal{X}_u^{e,k_3}(B_u^{d_{k_3}+2,k_3,h})^{-1}(C_u^{k_3})^+ dB_u^{d_{k_3}+2,k_3,h},
\end{align}
where $F^c$ is the cumulative interest of the margin account defined by
\begin{align}
\label{def:Fcumulativeinterest}
F_t^c := \int_0^t \mathcal{X}_u^{e,k_3} (B_u^{c,k_3,l})^{-1}(C_u^{k_3})^- dB_u^{c,k_3,l} - \int_0^t \mathcal{X}_u^{e,k_3} (B_u^{c,k_3,b})^{-1}(C_u^{k_3})^+ dB_u^{c,k_3,b}.
\end{align}
\end{definition}

A standard assumption consists in assuming that all cash accounts are absolutely continuous, so that all cash accounts can be written as $dB_t^j=r^j_tB^j_tdt$ for some $\GG$-adapted processes $r^j$ and any arbitrary index $j$ we consider in the present setting. When this is the case, one can simplify \eqref{def:Fcumulativeinterest} as
\begin{align*}
F_t^c := \int_0^t \mathcal{X}_u^{e,k_3}(C_u^{k_3})^- r_u^{c,k_3,l}du - \int_0^t \mathcal{X}_u^{e,k_3} (C_u^{k_3})^+ r_u^{c,k_3,b}du.
\end{align*}
The above formulation explicitly features the borrowing and lending collateral rates: the interest received on the posted collateral has a positive impact net of the interest paid on the received collateral.

We need the following generalization of the definition of a self-financing trading strategy.
\begin{definition}
\label{def:secondSelfFinancing}
Let $(x,\varphi,A^{k_2}, C^{k_3})$ be a collateralized trading strategy, where $k_2$ and $k_3$ fulfil the usual conditions. The strategy is called self financing, if the hedger's portfolio value $V^p(x,\varphi,A^{k_2},C^{k_3})$ fulfils
\begin{align}
\begin{aligned}
V_t^p(x,\varphi,A^{k_2},C^{k_3}):&=x+\sum_{k_1=1}^L\left\{\sum_{i=1}^{d_{k_1}}\left[\int_{(0, t]} \mathcal{X}_{u}^{e, k_{1}} \xi_{u}^{i, k_{1}}\left(d S_{u}^{i, k_{1}}+d D_{u}^{i, k_{1}}\right)\right.\right.\\
&\left.+\int_{(0, t]} \xi_{u}^{i, k_{1}} S_{u}^{i, k_{1}} d \mathcal{X}_{u}^{e, k_{1}}+\int_{(0, t]} \xi_{u}^{i, k_{1}} d\left[S^{i, k_{1}}, \mathcal{X}^{e, k_{1}}\right]_u\right]\\
&\left.+\sum_{j=0}^{d_{k_1}}\left[\int_{0}^{t} \mathcal{X}_{u}^{e, k_{1}} \psi_{u}^{j, k_{1}} d B_{u}^{j, k_{1}}+\textcolor{black}{\int_{(0,t]}} \psi_{u}^{j, k_{1}} B_{u}^{j, k_{1}} d \mathcal{X}_{u}^{e, k_{1}}\right]\right\}\\ 
&+ \int_{(0,t]}\mathcal{X}_u^{e,k_3}\xi_u^{d_{k_3}+1,k_3}\left(dS_u^{d_{k_3}+1,k_3} + dD_u^{d_{k_3}+1,k_3} \right) \\
&+\int_{(0,t]} \xi_u^{d_{k_3}+1,k_3}S_u^{d_{k_3}+1,k_3}d\mathcal{X}_u^{e,k_3}+  \int_{(0,t]} \xi_u^{d_{k_3}+1,k_3}d\left[S^{d_{k_3}+1,k_3},\mathcal{X}^{e,k_3} \right]_u \\
&+ \int_{(0,t]} \mathcal{X}_u^{e,k_3}\psi_u^{d_{k_3}+1,k_3} dB_u^{d_{k_3}+1,k_3} +\int_{(0,t]} \psi_u^{d_{k_3}+1,k_3}B_u^{d_{k_3}+1,k_3} d\mathcal{X}_u^{e,k_3}\\
&+\int_0^t \mathcal{X}_u^{e,k_3}\eta_u^{d_{k_3}+2,k_3} dB_u^{d_{k_3}+2,k_3}+\int_0^t \mathcal{X}_u^{e,k_3}\eta_u^{d_{k_3}+3,k_3} dB_u^{d_{k_3}+3,k_3}\\
&+F^c_t- \int_{(0,t]} C^{k_3}_ud\mathcal{X}_u^{e,k_3}+  A^{e,k_2}_t - V_t^{c}(x,\varphi,A^{k_2}, C^{k_3}),
\end{aligned}
\end{align}
for any $t \in [0,T].$
\end{definition}

Let us provide some information concerning the adjustment process and the rules for the determination of the amount of collateral. In line with \cite{BieRut15} the adjustment process satisfies $V_t^c(x,\varphi,A^{k_2},C^{k_3})=g(C^{k_3}_t(\varphi))$ for some typically Lipschitz function $g$. In the cases considered in the sequel we have either $V_t^c(x,\varphi,A^{k_2},C^{k_3})=\mathcal{X}^{e,k_3}_tC^{k_3}_t$ or $V_t^c(x,\varphi,A^{k_2},C^{k_3})=-\mathcal{X}^{e,k_3}_t(C^{k_3}_t)^-$.

The amount of collateral $C^{k_3}$ can be determined in many different ways, as the determination of such process is the result of a legal negotiation between the hedger and the counterparty. However it is rather common to link the collateral with the value (mark-to-market) of the contract.

\begin{remark}[Collateral and mark-to-market] We let $M$ be a $\GG$-adapted RCLL process that represents the value of the contract expressed in units of the local currency $e$. One possible specification for $M$ is given by the setting
\begin{align}
\label{eq:Mdef}
M_t:=V^0_t(x)-V_t(x,\varphi,A^{k_2},C^{k_3}).
\end{align}
The formulation above captures the natural assumption that the collateral amount is linked to the value of the contract from the perspective of the hedger. In particular, recall that the portfolio $V$ is meant to cover the liabilities of the hedger towards the counterparty, meaning that the market value of the contract is $-V$. In terms of the process $M$ one has the following specification for the collateral account under a generic currency $k_3$.
\begin{align}
\label{eq:collAccount}
\mathcal{X}^{e,k_3}_tC^{k_3}_t=(1+\delta^1_t)M^+_t-(1+\delta^2_t)M^-_t,
\end{align}
where the processes $\delta^1$ and $\delta^2$ represent haircuts that reduce/increase in percentage the amount of collateral. Using \eqref{eq:collAccount} and \eqref{eq:Mdef} we write
\begin{align*}
\mathcal{X}^{e,k_3}_tC^{k_3}_t=(1+\delta^1_t)\left(V^0_t(x)-V_t(x,\varphi,A^{k_2},C^{k_3})\right)^+-(1+\delta^2_t)\left(V^0_t(x)-V_t(x,\varphi,A^{k_2},C^{k_3})\right)^-.
\end{align*}
\end{remark}

\begin{remark}[Full collateralization] One particularly important case is the case known as \textit{full collateralization}. In this case the value of the collateral is continuously updated in time in order to perfectly match the value of the contract. This can be obtained by setting $\delta^1_t=\delta^2_t=0$ for every $t$, which gives
\begin{align}
\label{eq:fullCollat}
\mathcal{X}^{e,k_3}_tC^{k_3}_t=V^0_t(x)-V_t(x,\varphi,A^{k_2},C^{k_3}).
\end{align}
Finally, in the case where the initial endowment is zero we have
\begin{align}
\label{eq:fullCollat2}
\mathcal{X}^{e,k_3}_tC^{k_3}_t=-V_t(0,\varphi,A^{k_2},C^{k_3}).
\end{align}
An important fact to note is that when the transaction is fully collateralized, the collateralization scheme completely funds the trading portfolio of the hedger, so that the cost of the collateral (which is proportional to the collateral rate) coincides with the funding rate for the trading activity. 
\end{remark}

\subsection{Cash collateral} We first proceed to study the case where collateral is exchanged in cash as illustrated in Definition \ref{def:cashCollateral}. This constitutes the most common collateralization covenant. Cash collateral is the case that is also most commonly treated in the literature. From the present treatment we will be able to recover the findings of, among others, \cite{mopa17}, \cite{fushita10c} in the case of full collateralization. The risky asset used for collateralization is of course immaterial and in fact we shall set $\xi_t^{d_{k_3}+1,k_3}=0$ in the subsequent results.

\subsubsection{Margin account under segregation} Let us recall that under segregation, if the hedger receives collateral from the counterparty, he/she is not allowed to use it as a source of funding for the trading activity: this means that $(C^{k_3})^+$ (i.e. the received collateral) is immaterial in the hedger's wealth, only the collateral posted by the hedger $(C^{k_3})^-$ will have a role in the hedger's wealth. Concerning the received collateral $(C^{k_3})^+$, we notice that this loan, received from the counterparty, must be remunerated according to the cash account $B^{d_{k_3}+2,k_3,s}_t$, so that this remuneration will have an impact via the self-financing condition. On the other hand, the posted collateral $(C^{k_3})^-$ is borrowed from the account $B^{d_{k_3}+1,k_3}_t$ and is remunerated by the counterparty  with interest from the cash account $B^{c,k_3,l}$.

\begin{proposition}\label{prop:cashSegregation}
We assume the hedger operates under segregation, hence he/she is posting or receiving collateral in form of cash. Let $(x,\varphi,A^{k_2},C^{k_3})$ be a self-financing strategy and the following conditions hold for $t \in [0,T]$:
\begin{align}
\label{eq:cashcollsegr}
\begin{aligned}
&\xi_t^{d_{k_3}+1,k_3}= \psi_t^{d_{k_3}+1,k_3}=0,\\
&\eta_t^{d_{k_3}+3,k_3}= -(B_t^{d_{k_3}+3,k_3})^{-1}(C_t^{k_3})^-, \quad \eta_t^{d_{k_3}+2,k_3,s} = (B_t^{d_{k_3}+2,k_3,s})^{-1}(C_t^{k_3})^+.
\end{aligned}
\end{align}
Then hedger's wealth process $V(x,\varphi,A^{k_2},C^{k_3})$ is given by 
\begin{align}
\label{eq:vcashcollsegr}
\begin{aligned}
V_t(x,\varphi,A^{k_2},C^{k_3})&= V_t^p(x,\varphi,A^{k_2},C^{k_3}) + \mathcal{X}_t^{e,k_3}(C_t^{k_3})^-\\
&=\sum_{k_1=1}^L \mathcal{X}_t^{e,k_1} \left( \sum_{i=1}^{d_{k_1}} \xi_t^{i,k_1}S_t^{i,k_1} + \sum_{j=0}^{d_{k_1}} \psi_t^{j,k_1}B_t^{j,k_1}\right)\\
&\quad\quad+\mathcal{X}_t^{e,k_3}\eta_t^{d_{k_3}+3,k_3}B_t^{d_{k_3}+3,k_3}+ \mathcal{X}_t^{e,k_3}(C_t^{k_3})^-,
\end{aligned}
\end{align}
for every $t \in [0,T]$. In addition, the dynamics of the hedger's portfolio wealth are as follows for any $t \in [0,T]$ and $V^p(\varphi):= V^p(x,\varphi,A^{k_2},C^{k_3}) $ and $V_t^c(\varphi):= V_t^c(x,\varphi,A^{k_2},C^{k_3})$:
\begin{align}
\label{eq:cashSegregationVp}
\begin{aligned}
dV_t^p(\varphi) &= \tilde{V}_t^p(\varphi)dB_t^e + \sum_{k_1=1}^L  \sum_{i=1}^{d_{k_1}} \left[ \xi_t^{i,k_1}dK_t^{i,e,k_1} +\mathcal{X}_t^{e,k_1} \zeta_t^{i,k_1} (\tilde{B}_t^{i,k_1})^{-1} d\tilde{B}_t^{i,k_1} + \psi_t^{i,k_1}B_t^{i,k_1} d\mathcal{X}_t^{e,k_1} \right] \\
&+\sum_{k_1=1,k_1 \neq e}^L\psi_t^{0,k_1}B^e_td\left(\frac{B^{k_1}\mathcal{X}^{e,k_1}}{B^{e}}\right)_t \\
&- \mathcal{X}_t^{e,k_3} ({B}_t^{d_{k_3}+3,k_3})^{-1}(C_t^{k_3})^- d{B}_t^{d_{k_3}+3,k_3}+ \mathcal{X}_t^{e,k_3} ({B}_t^{d_{k_3}+2,k_3})^{-1}(C_t^{k_3})^+ d{B}_t^{d_{k_3}+2,k_3}\\ 
&+dF_t^c -C^{k_3}_td\mathcal{X}^{e,k_3}_t+ dA_t^{e,k_2} - dV_t^c(\varphi).
\end{aligned}
\end{align}
Under the repo constraint \eqref{eq:repoConstraint}, the dynamics of the hedger's wealth are given by
\begin{align}
\label{eq:VcashSegregation}
\begin{aligned}
dV_t(\varphi) &= \tilde{V}_t(\varphi)dB_t^e + \sum_{k_1=1}^L  \sum_{i=1}^{d_{k_1}} \xi_t^{i,k_1}\left[dK_t^{i,e,k_1}-S^{i,k_1}_td\mathcal{X}^{e,k_1}_t\right]\\
&+\sum_{k_1=1,k_1 \neq e}^L\psi_t^{0,k_1}B^e_td\left(\frac{B^{k_1}\mathcal{X}^{e,k_1}}{B^{e}}\right)_t  + dA_t^{e,k_2} +d\hat{F}_t^s,
\end{aligned}
\end{align}
where we use the notation $V(\varphi):= V(x,\varphi,A^{k_2},C^{k_3})$ and 
\begin{align}
\label{eq:hatFsegregation}
d\hat{F}_t^{s} := dF_t^s -C^{k_3}_td\mathcal{X}^{e,k_3}_t - \mathcal{X}_t^{e,k_3} ({B}_t^{d_{k_3}+3,k_3})^{-1}(C_t^{k_3})^- d{B}_t^{d_{k_3}+3,k_3}. 
\end{align}
\end{proposition}

\begin{proof}
By combining the assumptions made in \eqref{eq:cashcollsegr} with equality \eqref{eq:vcollwealth} and \eqref{borrowlendingcond}, we obtain
\begin{align*}
V_t(x,\varphi,A^{k_2},C^{k_3}) &= \sum_{k_1=1}^L \mathcal{X}_t^{e,k_1} \left( \sum_{i=1}^{d_{k_1}} \xi_t^{i,k_1}S_t^{i,k_1} + \sum_{j=0}^{d_{k_1}} \psi_t^{j,k_1}B_t^{j,k_1}\right)\\
&\quad\quad+ \mathcal{X}_t^{e,k_3}\left( -C_t^{k_3}+(C_t^{k_3})^++\eta_t^{d_{k_3}+3,k_3}B_t^{d_{k_3}+3,k_3} \right)\\
&= V_t^{p}(x,\varphi,A^{k_2},C^{k_3}) + \mathcal{X}_t^{e,k_3} (C_t^{k_3})^-,\\
\end{align*}
which proves \eqref{eq:vcashcollsegr}.
If we take a closer look at the hedger's portfolio value $V^p(x,\varphi,A^{k_2},C^{k_3})$ and recall that for some $k_1 = 1,\ldots,L$ we have $k_1=e$ and hence $\mathcal{X}^{e,e}\equiv 1$, $B^{0,e}:=B^e$, we have that
\begin{align*}
V_t^p(x,\varphi,A^{k_2},C^{k_3})&= \sum_{k_1=1}^L \mathcal{X}_t^{e,k_1} \sum_{i=1}^{d_{k_1}} \zeta_t^{i,k_1} + \sum_{k_1=1,k_1 \neq e}^L \mathcal{X}_t^{e,k_1}\psi_t^{0,k_1}B_t^{k_1} + \psi_t^{0,e}B_t^e - \mathcal{X}_t^{e,k_3} (C_t^{k_3})^-,
\end{align*}
for any $t \in [0,T]$, where the quantity $\zeta_t^{i,k_1}$ was defined in \eqref{eq:repoConstraint}. Hence we get
\begin{align}
\label{eq:vpstrategy}
\begin{aligned}
\psi_t^{0,e}&= \tilde{V}_t^p(x,\varphi,A^{k_2},C^{k_3}) - \sum_{k_1=1}^L \mathcal{X}_t^{e,k_1} \sum_{i=1}^{d_{k_1}} (B_t^e)^{-1}\zeta_t^{i,k_1}\\
& \quad- \sum_{k_1=1,k_1 \neq e}^L \mathcal{X}_t^{e,k_1}\psi_t^{0,k_1}(B_t^e)^{-1}B_t^{k_1} + \mathcal{X}_t^{e,k_3} (B_t^e)^{-1}(C_t^{k_3})^-,
\end{aligned}
\end{align}
with $\tilde{V}^p(x,\varphi,A^{k_2},C^{k_3}):= (B^e)^{-1}V^p(x,\varphi,A^{k_2},C^{k_3})$. By using the dynamics of the self financing condition \eqref{def:secondSelfFinancing} in combination with \eqref{eq:vpstrategy}, the notation $V^p(x,\varphi,A^{k_2},C^{k_3}):=V^{p}(\varphi)$, $V^c(x,\varphi,A^{k_2},C^{k_3}):=V^{c}(\varphi)$ and $V(x,\varphi,A^{k_2},C^{k_3}):=V(\varphi)$, by using 
\begin{align}
\begin{aligned}
\label{eq:dynamicsK}
dK_t^{i,e,k_1}&= B_t^{i,k_1}d\hat{S}_t^{i,cld,e,k_1} = B_t^{i,k_1} d\left(\hat{S}_t^{i,k_1}\mathcal{X}_t^ {e,k_1}\right) + (B_t^{i,k_1})^{-1}\mathcal{X}_t^{e,k_1}dD_t^{i,k_1}\\
&=\mathcal{X}_t^{e,k_1}(dS_t^{i,k_1} + dD_t^{i,k_1}) - (B_t^{i,k_1})^{-1}\mathcal{X}_t^{e,k_1} S_t^{i,k_1}dB_t^{i,k_1} + S_t^{i,k_1}d\mathcal{X}_t^{e,k_1} + d\left[S^{i,k_1},\mathcal{X}^{e,k_1} \right]_t,
\end{aligned}
\end{align}
from equation \eqref{eq:definitionKi} for $i = 1,\ldots, d_{k_1}$, $k_1=1,\ldots, L$, we have that
\begin{align*}
dV_t^p(\varphi)&= \sum\limits_{k_1=1}^L \sum_{i=1}^{d_{k_1}} \left[\mathcal{X}_t^{e,k_1}\xi_t^{i,k_1}(dS_t^{i,k_1}+dD_t^{i,k_1})+\xi_t^{i,k_1} S_t^{i,k_1}d\mathcal{X}_t^{e,k_1} +\xi_t^{i,k_1}d\left[ S^{i,k_1},\mathcal{X}^{e,k_1}\right]_t \right. \\
&- \xi_t^{i,k_1}(B_t^{i,k_1})^{-1}\mathcal{X}_t^{e,k_1} S_t^{i,k_1}dB_t^{i,k_1} + \underbrace{\xi_t^{i,k_1}(B_t^{i,k_1})^{-1}\mathcal{X}_t^{e,k_1} S_t^{i,k_1}dB_t^{i,k_1} + \mathcal{X}_t^{e,k_1}\psi_t^{i,k_1}dB_t^{i,k_1}}_{= (B_t^{i,k_1})^{-1}\mathcal{X}_t^{e,k_1} \zeta_t^{i,k_1} dB_t^{i,k_1}} \\
&\left.+ \psi_t^{i,k_1}B_t^{i,k_1}d\mathcal{X}_t^{e,k_1} \right] +\sum_{k_1=1, k_1\neq
e}^L \left[ \mathcal{X}_t^{e,k_1} \psi_t^{0,k_1}dB_t^{k_1} + \psi_t^{0,k_1}B_t^{k_1}d\mathcal{X}_t^{e,k_1} \right] \\
&+ \left(\tilde{V}_t^p(\varphi) - \sum_{k_1=1}^L \mathcal{X}_t^{e,k_1} \sum_{i=1}^{d_{k_1}} (B_t^e)^{-1}\zeta_t^{i,k_1} - \sum_{k_1=1,k_1 \neq e}^L \mathcal{X}_t^{e,k_1}\psi_t^{0,k_1}(B_t^e)^{-1}B_t^{k_1} +\mathcal{X}_t^{e,k_3} (B_t^e)^{-1} (C_t^{k_3})^- \right) dB_t^e \\
&- \mathcal{X}_t^{e,k_3} (B_t^{d_{k_3}+3,k_3})^{-1} (C_t^{k_3})^- dB_t^{d_{k_3}+3,k_3} + \mathcal{X}_t^{e,k_3} ({B}_t^{d_{k_3}+2,k_3})^{-1}(C_t^{k_3})^+ d{B}_t^{d_{k_3}+2,k_3}\\
& + dF_t^c  -C_t^{k_3} d\mathcal{X}_t^{e,k_3}+ dA_t^{e,k_2} - dV_t^c(\varphi) \\
&= \tilde{V}_t^p(\varphi)dB_t^e + \sum_{k_1=1}^L  \sum_{i=1}^{d_{k_1}} \left[ \xi_t^{i,k_1}dK_t^{i,e,k_1} +\mathcal{X}_t^{e,k_1} \zeta_t^{i,k_1} \underbrace{((B_t^{i,k_1})^{-1}dB_t^{i,k_1} - (B_t^e)^{-1}dB_t^e)}_{=(\tilde{B}_t^{i,k_1})^{-1} d\tilde{B}_t^{i,k_1}} + \psi_t^{i,k_1}B_t^{i,k_1} d\mathcal{X}_t^{e,k_1} \right] \\
&+ \sum_{k_1=1,k_1 \neq e}^L \left( \mathcal{X}_t^{e,k_1} \psi_t^{0,k_1} B_t^{k_1}\underbrace{((B_t^{k_1})^{-1}dB_t^{k_1} - (B_t^e)^{-1}dB_t^e )}_{= (\tilde{B}_t^{k_1})^{-1} d\tilde{B}_t^{k_1}} + \psi_t^{0,k_1} B_t^{k_1}d\mathcal{X}_t^{e,k_1} \right)+\mathcal{X}_t^{e,k_3} (B_t^e)^{-1} (C_t^{k_3})^-dB_t^e \\
&- \mathcal{X}_t^{e,k_3} (B_t^{d_{k_3}+3,k_3})^{-1} (C_t^{k_3})^- dB_t^{d_{k_3}+3,k_3} + \mathcal{X}_t^{e,k_3} ({B}_t^{d_{k_3}+2,k_3})^{-1}(C_t^{k_3})^+ d{B}_t^{d_{k_3}+2,k_3}\\
& + dF_t^c  -C_t^{k_3} d\mathcal{X}_t^{e,k_3}+ dA_t^{e,k_2} - dV_t^c(\varphi)\\
&= \tilde{V}_t^p(\varphi)dB_t^e + \sum_{k_1=1}^L  \sum_{i=1}^{d_{k_1}} \left[ \xi_t^{i,k_1}dK_t^{i,e,k_1} +\mathcal{X}_t^{e,k_1} \zeta_t^{i,k_1} (\tilde{B}_t^{i,k_1})^{-1} d\tilde{B}_t^{i,k_1} + \psi_t^{i,k_1}B_t^{i,k_1} d\mathcal{X}_t^{e,k_1} \right] \\
&+ \sum_{k_1=1,k_1 \neq e}^L \left( \mathcal{X}_t^{e,k_1} \psi_t^{0,k_1} B_t^{k_1} (\tilde{B}_t^{k_1})^{-1} d\tilde{B}_t^{k_1} + \psi_t^{0,k_1} B_t^{k_1}d\mathcal{X}_t^{e,k_1} \right) +\mathcal{X}_t^{e,k_3} (B_t^e)^{-1} (C_t^{k_3})^-dB_t^e\\
&- \mathcal{X}_t^{e,k_3} ({B}_t^{d_{k_3}+1,k_3})^{-1}(C_t^{k_3})^- d{B}_t^{d_{k_3}+1,k_3} +dF_t^s  -C_t^{k_3} d\mathcal{X}_t^{e,k_3}+ dA_t^{e,k_2} - dV_t^c(\varphi).
\end{align*}
We obtain \eqref{eq:cashSegregationVp} by noticing that we can perform the following simplification
\begin{align*}
 \sum_{k_1=1,k_1 \neq e}^L \left( \mathcal{X}_t^{e,k_1} \psi_t^{0,k_1} B_t^{k_1} (\tilde{B}_t^{k_1})^{-1} d\tilde{B}_t^{k_1} + \psi_t^{0,k_1} B_t^{k_1}d\mathcal{X}_t^{e,k_1} \right)= \sum_{k_1=1,k_1 \neq e}^L\psi_t^{0,k_1}B^e_td\left(\frac{B^{k_1}\mathcal{X}^{e,k_1}}{B^{e}}\right)_t.
\end{align*}
Furthermore, if condition \eqref{eq:repoConstraint} holds, meaning that $\zeta^{i,k_1}_t = 0$ for all $t\in[0,T]$ $i=1,\ldots,d_{k_1}$, $k_1=1,\ldots,L$ and so $\psi^{i,k_1}_tB^{i,k_1}_t=-\xi^{i,k_1}_tS^{i,k_1}_t$, then the dynamics of the hedger's wealth process are given by
\begin{align*}
dV_t(\varphi) &= dV_t^p(\varphi)+dV_t^c(\varphi) \\
&= \tilde{V}_t(\varphi)dB_t^e + \sum_{k_1=1}^L  \sum_{i=1}^{d_{k_1}} \xi_t^{i,k_1}\left(dK_t^{i,e,k_1}-S^{i,k_1}_td\mathcal{X}^{e,k_1}_t\right)+\sum_{k_1=1,k_1 \neq e}^L\psi_t^{0,k_1}B^e_td\left(\frac{B^{k_1}\mathcal{X}^{e,k_1}}{B^{e}}\right)_t\\
&+ dA_t^{e,k_2} +dF_t^s - \mathcal{X}_t^{e,k_3} ({B}_t^{d_{k_3}+1,k_3})^{-1}(C_t^{k_3})^- d{B}_t^{d_{k_3}+1,k_3} -C_t^{k_3} d\mathcal{X}_t^{e,k_3},
\end{align*}
where we also used $\tilde{V}_t^p(\varphi)=\tilde{V}_t(\varphi)-\mathcal{X}_t^{e,k_3} (B_t^e)^{-1} (C_t^{k_3})^-$.
\end{proof}

\subsubsection{Margin account under rehypothecation} Let us recall that, under rehypothecation, when the hedger receives the collateral amount $\left(C^{k_3}\right)^+$ he/she  can use it to fund his/her trading activity. Interest is paid by the hedger to the counterparty based on $\left(C^{k_3}\right)^+$ and the cash account $B^{c,k_3,b}$.  Instead, in case the hedger posts the amount $\left(C^{k_3}\right)^-$ to the counterparty, then the hedger will receive from the counterparty an interest amount based on $\left(C^{k_3}\right)^-$  and the cash account $B^{c,k_3,l}$. As the hedger needs to rise the amount of cash $\left(C^{k_3}\right)^-$ he/she borrows such amount from the dedicated cash account $B^{d_{k_3}+1,k_3}$ that might coincide with the unsecured cash account in currency $k_3$ i.e. $B^{k_3}$.

The present case is the most common one in the market for bilateral trades (i.e. trades not involving a central counterparty) and, when the collateral is perfect (as in \eqref{eq:fullCollat}) then we will obtain in the sequel useful valuation formulas based on the present case.

\begin{proposition}\label{prop:cashRehypothecation}
Consider the market model, where the hedger delivers or posts collateral in form of cash under rehypothecation. Let $(x,\varphi,A^{k_2},C^{k_3})$ be a self financing trading strategy and the following conditions hold for any $t \in [0,T]$:
\begin{align}
\label{eq:cashcollrehy}
\xi_t^{d_{k_3}+1,k_3}= \psi_t^{d_{k_3}+1,k_3} = 0, \quad \eta_t^{d_{k_3}+3,k_3} = - (B_t^{d_{k_3}+3,k_3})^{-1}(C_t^{k_3})^-, \quad \eta_t^{d_{k_3}+2,k_3,h}=0.
\end{align}
Consequently, the hedger's wealth process $V(x,\varphi,A^{k_2},C^{k_3})$ is given by
\begin{align}
\label{eq:vcashcollrehy}
\begin{aligned}
V_t(x,\varphi,A^{k_2},C^{k_3}) &= \sum_{k_1=1}^L \mathcal{X}_t^{e,k_1} \left( \sum_{i=1}^{d_{k_1}} \xi_t^{i,k_1}S_t^{i,k_1} + \sum_{j=0}^{d_{k_1}} \psi_t^{j,k_1}B_t^{j,k_1}\right)- \mathcal{X}_t^{e,k_3}(C_t^{k_3})^+\\
&= V_t^p(x,\varphi,A^{k_2},C^{k_3}) - \mathcal{X}_t^{e,k_3} C_t^{k_3}.
\end{aligned}
\end{align}
and the dynamics of the hedger's portfolio value $V_t^p(\varphi):= V_t^p(x,\varphi,A^{k_2},C^{k_3})$ are
\begin{align}
\begin{aligned}
dV_t^p(\varphi)&= \tilde{V}_t^p(\varphi)dB_t^e + \sum_{k_1=1}^L \sum_{i=1}^{d_{k_1}} \left[ \xi_t^{i,k_1}dK_t^{i,e,k_1} + \mathcal{X}_t^{e,k_1} \zeta_t^{i,k_1} (\tilde{B}_t^{i,k_1})^{-1} d\tilde{B}_t^{i,k_1} + \psi_t^{i,k_1} B_t^{i,k_1} d\mathcal{X}_t^{e,k_1} \right] \\
&+\sum_{k_1=1,k_1 \neq e}^L\psi_t^{0,k_1}B^e_td\left(\frac{B^{k_1}\mathcal{X}^{e,k_1}}{B^{e}}\right)_t \\
&- \mathcal{X}_t^{e,k_3} (B_t^{d_{k_3}+3,k_3})^{-1} (C_t^{k_3})^- dB_t^{d_{k_3}+3,k_3}+ \mathcal{X}_t^{e,k_3} (B_t^{e})^{-1} (C_t^{k_3})^+ dB_t^{e} \\
& + dF_t^c- C_t^{k_3} d\mathcal{X}_t^{e,k_3} + dA_t^{e,k_2} - dV_t^c(\varphi).
\end{aligned}
\end{align}
Hence the dynamics of the hedger's wealth process under the repo constraint \eqref{eq:repoConstraint} can be denoted by
\begin{align}
\label{eq:VcashRehypothecation}
\begin{aligned}
dV_t(\varphi) &= \tilde{V}_t(\varphi)dB_t^e + \sum_{k_1=1}^L  \sum_{i=1}^{d_{k_1}} \xi_t^{i,k_1}\left[dK_t^{i,e,k_1}-S^{i,k_1}_td\mathcal{X}^{e,k_1}_t\right]\\
&+\sum_{k_1=1,k_1 \neq e}^L\psi_t^{0,k_1}B^e_td\left(\frac{B^{k_1}\mathcal{X}^{e,k_1}}{B^{e}}\right)_t  + dA_t^{e,k_2} +d\hat{F}_t^h,
\end{aligned}
\end{align}
where 
\begin{align}
\label{eq:hatFrehypothecation}
d\hat{F}_t^h = dF_t^c- C_t^{k_3} d\mathcal{X}_t^{e,k_3} - \mathcal{X}_t^{e,k_3} (B_t^{d_{k_3}+3,k_3})^{-1} (C_t^{k_3})^- dB_t^{d_{k_3}+3,k_3}+ \mathcal{X}_t^{e,k_3} (B_t^{e})^{-1} (C_t^{k_3})^+ dB_t^{e}.
\end{align}
\end{proposition}

\begin{proof}
By using the assumptions \eqref{eq:cashcollrehy} combined with \eqref{borrowlendingcond} and \eqref{eq:vcollwealth}, we get for any $t \in [0,T]$
\begin{align*}
V_t(x,\varphi,A^{k_2},C^{k_3}) &= \sum_{k_1=1}^L \mathcal{X}_t^{e,k_1} \left( \sum_{i=1}^{d_{k_1}} \xi_t^{i,k_1}S_t^{i,k_1} + \sum_{j=0}^{d_{k_1}} \psi_t^{j,k_1}B_t^{j,k_1}\right)+ \mathcal{X}_t^{e,k_3}\left( -(C_t^{k_3})^- -C_t^{k_3}\right) \\
&= \sum_{k_1=1}^L \mathcal{X}_t^{e,k_1} \left( \sum_{i=1}^{d_{k_1}} \xi_t^{i,k_1}S_t^{i,k_1} + \sum_{j=0}^{d_{k_1}} \psi_t^{j,k_1}B_t^{j,k_1}\right)- \mathcal{X}_t^{e,k_3}(C_t^{k_3})^+\\
&= V_t^p(x,\varphi,A^{k_2},C^{k_3}) - \mathcal{X}_t^{e,k_3} C_t^{k_3}.
\end{align*}
Hence the hedger's portfolio wealth gives us
\begin{align}
\begin{aligned}
\label{eq:vpstrategy2}
\psi_t^{0,e} &= \tilde{V}_t^p(x,\varphi,A^{k_2},C^{k_3}) - \sum_{k_1=1}^L \mathcal{X}_t^{e,k_1} \sum_{i=1}^{d_{k_1}}(B^e_t)^{-1} \zeta_t^{i,k_1}\\
&\quad - \sum_{k_1=1, k_1 \neq e}^L \mathcal{X}_t^{e,k_1}\psi_t^{0,k_1}(B_t^e)^{-1}B_t^{k_1} + \mathcal{X}_t^{e,k_3} (B_t^e)^{-1}(C_t^{k_3})^-,
\end{aligned}
\end{align}
and combining \eqref{eq:vpstrategy2} with the self financing condition \eqref{def:secondSelfFinancing} and \eqref{eq:dynamicsK}, we receive
\begin{align*}
dV_t^p(\varphi) &= \tilde{V}_t^p(\varphi)dB_t^e + \sum_{k_1=1}^L \sum_{i=1}^{d_{k_1}} \left[ \xi_t^{i,k_1}dK_t^{i,e,k_1} + \mathcal{X}_t^{e,k_1} \zeta_t^{i,k_1} (\tilde{B}_t^{i,k_1})^{-1} d\tilde{B}_t^{i,k_1} + \psi_t^{i,k_1} B_t^{i,k_1} d\mathcal{X}_t^{e,k_1} \right] \\
&+ \sum_{k_1=1,k_1 \neq e}^L \left( \mathcal{X}_t^{e,k_1} \psi_t^{0,k_1} B_t^{k_1} (\tilde{B}_t^{k_1})^{-1} d\tilde{B}_t^{k_1} + \psi_t^{0,k_1} B_t^{k_1}d\mathcal{X}_t^{e,k_1} \right) \\
&- \mathcal{X}_t^{e,k_3} (B_t^{d_{k_3}+3,k_3})^{-1} (C_t^{k_3})^- dB_t^{d_{k_3}+3,k_3}+\mathcal{X}_t^{e,k_3}(B_t^e)^{-1} (C_t^{k_3})^- dB_t^e- C_t^{k_3}d\mathcal{X}_t^{e,k_3} \\
& + dF_t^c + dA_t^{e,k_2} - dV_t^c(\varphi),
\end{align*}
where $V(\varphi),\, V^p(\varphi)$ and $V^c(\varphi)$ are defined as before. In addition, let the repo constraint \eqref{eq:repoConstraint} be fulfilled and the dynamics of the hedger's wealth process $V(\varphi)$ are given by
\begin{align*}
dV_t(\varphi) &= \tilde{V}_t^p(\varphi)dB_t^e + \sum_{k_1=1}^L \sum_{i=1}^{d_{k_1}}  \xi_t^{i,k_1}\left[dK_t^{i,e,k_1}-S^{i,k_1}_td\mathcal{X}^{e,k_1}_t\right] \\
&+\sum_{k_1=1,k_1 \neq e}^L\psi_t^{0,k_1}B^e_td\left(\frac{B^{k_1}\mathcal{X}^{e,k_1}}{B^{e}}\right)_t\\
&- \mathcal{X}_t^{e,k_3} (B_t^{d_{k_3}+3,k_3})^{-1} (C_t^{k_3})^- dB_t^{d_{k_3}+3,k_3}+\mathcal{X}_t^{e,k_3}(B_t^e)^{-1} (C_t^{k_3})^- dB_t^e- C_t^{k_3}d\mathcal{X}_t^{e,k_3} \\
& + dF_t^c + dA_t^{e,k_2} .
\end{align*}
By observing that $\tilde{V}_t^p(\varphi)=\tilde{V}_t(\varphi)+\mathcal{X}_t^{e,k_3}(B_t^e)^{-1} (C_t^{k_3})^+-\mathcal{X}_t^{e,k_3}(B_t^e)^{-1} (C_t^{k_3})^-$ we conclude.
\end{proof}

\subsection{Risky asset collateral}

Formally, there is no need to distinguish between the case where the hedger posts or receives collateral in form of shares of the risky asset $S^{d_{k_3}+1,k_3}$ under segregation or rehypothecation, since the hedger's wealth process behaves in the the same way modulo the different reinvestment rates $B^{d_{k_3}+2,k_3,s}$ and respectively $B^{d_{k_3}+2,k_3,h}$. In the following, the index $h$ can be replaced by $s$ without loss of generality to formally make a distinction between the underlying collateral conventions.

\subsubsection{Risky asset collateral under segregation and rehypothecation}

\begin{proposition}\label{prop:RiskyCollateral}
Consider the hedger posting or receiving collateral in form of shares of the risky asset $S^{d_{k_3}+1,k_3}$ with no further restrictions concerning the underlying collateral conventions. Let $(x,\varphi,A^{k_2},C^{k_3})$ be a self financing trading strategy and assume that the following conditions hold for $t \in [0,T]$:
\begin{align}
\label{eq:riskycoll}
\begin{aligned}
&\xi_t^{d_{k_3}+1,k_3} =(S_t^{d_{k_3}+1,k_3})^{-1}(C_t^{k_3})^{-}, \quad \psi_t^{d_{k_3}+1,k_3} = -(B_t^{d_{k_3}+1,k_3})^{-1} (C_t^{k_3})^-,\\ &\eta_t^{d_{k_3}+2,k_3,h} = (B_t^{d_{k_3}+2,k_3,h})^{-1} (C_t^{k_3})^+,\quad \eta_t^{d_{k_3}+3,k_3} = 0.
\end{aligned}
\end{align}
The hedger's wealth process is now given by 
\begin{align}
\label{eq:vriskycoll}
\begin{aligned}
V_t(x,\varphi,A^{k_2},C^{k_3}) &= \sum_{k_1=1}^L \mathcal{X}_t^{e,k_1} \left(\sum_{i=1}^{d_{k_1}}\xi_t^{i,k_1}S_t^{i,k_1} + \sum_{j=0}^{d_{k_1}} \psi_t^{j,k_1}B_t^{j,k_1}\right) + \mathcal{X}_t^{e,k_3}(C_t^{k_3})^- \\
&= V_t^p(x,\varphi,A^{k_2}, C^{k_3}) + \mathcal{X}_t^{e,k_3}(C_t^{k_3})^-
\end{aligned}
\end{align}
and $V^c_t(\varphi)=\mathcal{X}_t^{e,k_3}(C_t^{k_3})^-$. The dynamics of the hedger's portfolio value $V_t^p(\varphi):= V_t^p(x,\varphi,A^{k_2},C^{k_3})$ are
\begin{align}
\begin{aligned}
dV_t^p(\varphi)&=\tilde{V}_t^p(\varphi)dB_t^e + \sum_{k_1=1}^L \sum_{i=1}^{d_{k_1}} \left[ \xi_t^{i,k_1}dK_t^{i,e,k_1} + \mathcal{X}_t^{e,k_1} \zeta_t^{i,k_1} (\tilde{B}_t^{i,k_1})^{-1} d\tilde{B}_t^{i,k_1} + \psi_t^{i,k_1} B_t^{i,k_1} d\mathcal{X}_t^{e,k_1} \right] \\
&+\sum_{k_1=1,k_1 \neq e}^L\psi_t^{0,k_1}B^e_td\left(\frac{B^{k_1}\mathcal{X}^{e,k_1}}{B^{e}}\right)_t + \xi_t^{d_{k_3}+1,k_3}dK_t^{d_{k_3}+1,e,k_3}+ \psi_t^{d_{k_3}+1,k_3}B_t^{d_{k_3}+1,k_3}d\mathcal{X}_t^{e,k_3}\\
&+\mathcal{X}_t^{e,k_3}(B_t^{d_{k_3}+2,k_3,h})^{-1} (C_t^{k_3})^+dB_t^{d_{k_3}+2,k_3,h}+dF_t^c-C_t^{k_3}d\mathcal{X}_t^{e,k_3} + dA_t^{e,k_2} - dV_t^c(\varphi). 
\end{aligned}
\end{align}
It follows that the dynamics of the hedger's wealth process are given by
\begin{align}
\begin{aligned}
\label{eq:dVriskycoll}
dV_t(\varphi) &=\tilde{V}_t(\varphi)dB_t^e + \sum_{k_1=1}^L \sum_{i=1}^{d_{k_1}} \xi_t^{i,k_1}\left[ dK_t^{i,e,k_1} -S_t^{i,k_1} d\mathcal{X}_t^{e,k_1} \right] \\
&+\sum_{k_1=1,k_1 \neq e}^L\psi_t^{0,k_1}B^e_td\left(\frac{B^{k_1}\mathcal{X}^{e,k_1}}{B^{e}}\right)_t + \xi_t^{d_{k_3}+1,k_3}\left[dK_t^{d_{k_3}+1,e,k_3}-S_t^{d_{k_3}+1,k_3}d\mathcal{X}_t^{e,k_3}\right]\\
&+ dA_t^{e,k_2} +d\bar{F}_t^h, 
\end{aligned}
\end{align} 
under the repo constraint \eqref{eq:repoConstraint} with
\begin{align}
\label{eq:barFriskyassetcoll}
d\bar{F}_t^h =dF_t^c- C_t^{k_3} d\mathcal{X}_t^{e,k_3}+\mathcal{X}_t^{e,k_3}(B_t^{d_{k_3}+2,k_3,h})^{-1} (C_t^{k_3})^+dB_t^{d_{k_3}+2,k_3,h} - \mathcal{X}_t^{e,k_3} (B_t^{e})^{-1} (C_t^{k_3})^- dB_t^{e}.
\end{align}
\end{proposition}

\begin{proof}
Combining assumption \eqref{eq:riskycoll} with \eqref{borrowlendingcond} and \eqref{eq:vcollwealth}, we receive
\begin{align*}
V_t(x,\varphi,A^{k_2},C^{k_3}) &= \sum_{k_1=1}^L \mathcal{X}_t^{e,k_1} \left(\sum_{i=1}^{d_{k_1}}\xi_t^{i,k_1}S_t^{i,k_1} + \sum_{j=0}^{d_{k_1}} \psi_t^{j,k_1}B_t^{j,k_1}\right) \\
&+ \mathcal{X}_t^{e,k_3}\left( \underbrace{ (C_t^{k_3})^- - (C_t^{k_3})^- - C_t^{k_3} + (C_t^{k_3})^+ }_{= (C_t^{k_3})^-}\right) \\
&= V_t^p(x,\varphi,A^{k_2}, C^{k_3}) + \mathcal{X}_t^{e,k_3}(C_t^{k_3})^-
\end{align*}
and thus \eqref{eq:vriskycoll} for any $t \in [0,T]$. By using
\begin{align}
\psi_t^{0,e} = \tilde{V}_t^p(x,\varphi,A^{k_2},C^{k_3}) - \sum_{k_1=1}^L \mathcal{X}_t^{e,k_1} \sum_{i=1}^{d_{k_1}}(B^e_t )^{-1}\zeta_t^{i,k_1} -  \sum_{k_1=1, k_1 \neq e}^L \mathcal{X}_t^{e,k_1} \psi_t^{0,k_1} (B_t^e)^{-1} B_t^{k_1},
\end{align}
the self financing condition \eqref{def:secondSelfFinancing} and the dynamics of $K^{d_{k_3}+1,e,k_3}$ given by \eqref{eq:dynamicsK}, the dynamics of the hedger's portfolio wealth $V^p(x,\varphi,A^{k_2},C^{k_3}):= V^p(\varphi)$ are given by
\begin{align*}
dV_t^p(\varphi)&=\tilde{V}_t^p(\varphi)dB_t^e + \sum_{k_1=1}^L \sum_{i=1}^{d_{k_1}} \left[ \xi_t^{i,k_1}dK_t^{i,e,k_1} + \mathcal{X}_t^{e,k_1} \zeta_t^{i,k_1} (\tilde{B}_t^{i,k_1})^{-1} d\tilde{B}_t^{i,k_1} + \psi_t^{i,k_1} B_t^{i,k_1} d\mathcal{X}_t^{e,k_1} \right] \\
&+\sum_{k_1=1,k_1 \neq e}^L\psi_t^{0,k_1}B^e_td\left(\frac{B^{k_1}\mathcal{X}^{e,k_1}}{B^{e}}\right)_t\\
&+ \xi_t^{d_{k_3}+1,k_3} \left(\mathcal{X}_t^{e,k_3} (dS_t^{d_{k_3}+1,k_3} + dD_t^{d_{k_3}+1,k_3}) +d\left[ S^{d_{k_3}+1,k_3}, \mathcal{X}^{e,k_3}  \right]_t + S_t^{d_{k_3}+1,k_3} d\mathcal{X}_t^{e,k_1}\right. \\
& \quad\quad\quad\quad\quad \left. - (B_t^{d_{k_3}+1,k_3})^{-1} \mathcal{X}_t^{e,k_3} S_t^{d_{k_3}+1,k_3} dB_t^{d_{k_3}+1,k_3} \right) + \psi_t^{d_{k_3}+1,k_3}B_t^{d_{k_3}+1,k_3}d\mathcal{X}_t^{e,k_3} \\
&+\underbrace{\xi_t^{d_{k_3}+1,k_3} (B_t^{d_{k_3}+1,k_3})^{-1}\mathcal{X}_t^{e,k_3} S_t^{d_{k_3}+1,k_3} dB_t^{d_{k_3}+1,k_3} + \mathcal{X}_t^{e,k_3} \psi_t^{d_{k_3}+1,k_3} dB_t^{d_{k_3}+1,k_3}}_{= \mathcal{X}_t^{e,k_3} \left( B_t^{d_{k_3}+1,k_3}\right)^{-1}  \underbrace{ \left(\xi_t^{d_{k_3}+1,k_3}S_t^{d_{k_3}+1,k_3} + \psi_t^{d_{k_3}+1,k_3} B_t^{d_{k_3}+1,k_3} \right) }_{\overset{\eqref{eq:riskycoll}}{=} 0} dB_t^{d_{k_3}+1,k_3} }\\
&+\mathcal{X}_t^{e,k_3}(B_t^{d_{k_3}+2,k_3,h})^{-1} (C_t^{k_3})^+dB_t^{d_{k_3}+2,k_3,h}+dF_t^c-C_t^{k_3}d\mathcal{X}_t^{e,k_3} + dA_t^{e,k_2} - dV_t^c(\varphi) \\
&=\tilde{V}_t^p(\varphi)dB_t^e + \sum_{k_1=1}^L \sum_{i=1}^{d_{k_1}} \left[ \xi_t^{i,k_1}dK_t^{i,e,k_1} + \mathcal{X}_t^{e,k_1} \zeta_t^{i,k_1} (\tilde{B}_t^{i,k_1})^{-1} d\tilde{B}_t^{i,k_1} + \psi_t^{i,k_1} B_t^{i,k_1} d\mathcal{X}_t^{e,k_1} \right] \\
&+\sum_{k_1=1,k_1 \neq e}^L\psi_t^{0,k_1}B^e_td\left(\frac{B^{k_1}\mathcal{X}^{e,k_1}}{B^{e}}\right)_t + \xi_t^{d_{k_3}+1,k_3}dK_t^{d_{k_3}+1,e,k_3}+ \psi_t^{d_{k_3}+1,k_3}B_t^{d_{k_3}+1,k_3}d\mathcal{X}_t^{e,k_3}\\
&+\mathcal{X}_t^{e,k_3}(B_t^{d_{k_3}+2,k_3,h})^{-1} (C_t^{k_3})^+dB_t^{d_{k_3}+2,k_3,h}+dF_t^c-C_t^{k_3}d\mathcal{X}_t^{e,k_3} + dA_t^{e,k_2} - dV_t^c(\varphi). 
\end{align*}
Let the repo constraint \eqref{eq:repoConstraint} be fulfilled. Hence the dynamics of the wealth process $V(\varphi):=V(x,\varphi,A^{k_2},C^{k_3})$ are given by
\begin{align*}
dV_t(\varphi) &=\tilde{V}_t^p(\varphi)dB_t^e + \sum_{k_1=1}^L \sum_{i=1}^{d_{k_1}} \left[ \xi_t^{i,k_1}dK_t^{i,e,k_1} -S_t^{i,k_1} d\mathcal{X}_t^{e,k_1} \right] \\
&+\sum_{k_1=1,k_1 \neq e}^L\psi_t^{0,k_1}B^e_td\left(\frac{B^{k_1}\mathcal{X}^{e,k_1}}{B^{e}}\right)_t + \xi_t^{d_{k_3}+1,k_3}\left[dK_t^{d_{k_3}+1,e,k_3}-S_t^{d_{k_3}+1,k_3}d\mathcal{X}_t^{e,k_3}\right]\\
&+\mathcal{X}_t^{e,k_3}(B_t^{d_{k_3}+2,k_3,h})^{-1} (C_t^{k_3})^+dB_t^{d_{k_3}+2,k_3,h}+dF_t^c-C_t^{k_3}d\mathcal{X}_t^{e,k_3} + dA_t^{e,k_2}. 
\end{align*}
Since we have $\tilde{V}_t^p(\varphi)=\tilde{V}_t(\varphi)-\mathcal{X}_t^{e,k_3}(B^e_t)^{-1}(C^{k_3}_t)^-$ we conclude.
\end{proof}

\section{Pricing under funding costs and collateralization}\label{sec:pricingExogColl}

Pricing in the absence of collateralization was discussed in Section \ref{sec:pricingNoColl}, where we defined the hedger's fair price $\bar{p}^e$. In this section we want to show that pricing in a multi-currency setting can be processed similarly to Proposition 5.1 of \cite{BieRut15}. 

Pricing will be analysed from the perspective of the hedger: given the contractually agreed cumulative stream of cashflows  $A^{k_2}-A^{k_2}_0$, the objective of the hedger is to find $p^e_0=A^{k_2}_0$ by means of replication, i.e. by investing according to an admissible trading strategy.

\begin{definition}[\cite{BieRut15} Definition 5.1]
\label{def:replicating}
Let $t \in [0,T]$ and $p_t^e$ be a $\cG_t$-measurable random variable. A self financing trading strategy 
\begin{align}
\label{def:V0replicating}
(V_t^0(x) + p_t^e, \varphi, A^{k_2} - A_t^{k_2}, C^{k_3})
\end{align}
replicates the collateralized contract $(A^{k_2},C^{k_3})$ on the interval $[t,T]$ whenever
\begin{align}
\label{eq:V0replicating}
V_T(V_t^0(x) + p_t^e, \varphi, A^{k_2} - A_t^{k_2}, C^{k_3}) = V_T^0(x).
\end{align}
\end{definition}

\begin{definition}[\cite{BieRut15} Definition 5.2]
Let $t \in [0,T]$. Any $\cG_t$-measurable random variable $p_t^e$ for which there exists a replicating strategy for $(A^{k_2},C^{k_3})$ over $[t,T]$  is called ex-dividend price at time $t$ of the contract $A^{e,k_2}$ associated with $\varphi$, also denoted by $S_t(x,\varphi,A^{k_2},C^{k_3})$.
\end{definition}

The following points from \cite{BieRut15} are important: first, any ex-dividend price $p^e_0$ of $A^{e,k_2}$ is also a hedger's fair price $\bar{p}^e_0$ for $A^{e,k_2}$ at time $0$. Secondly, the ex-dividend price in general might depend on the initial endowment $x$ and the choice of $\varphi$. However, for the sake of the present treatment, the ex-dividend price will be independent of the choice of $x$ and $\varphi$ and equivalent to the valuation ex dividend price defined below. Recall from section \eqref{arbitrage1} that the future value of the hedger's initial endowment is given as $V_t^{0}(x) = xB_t^e$ for any $t \in [0,T]$.

\begin{definition}[\cite{BieRut15} Definition 5.3]
Assume that an admissible self-financing trading strategy $(x,\varphi,A^{k_2},C^{k_3})$ replicates $(A^{k_2},C^{k_3})$ on $[0,T]$. Then the process $\hat{p}^e_t:=V_t(x,\varphi,A^{k_2},C^{k_3})-V^0_t(x)$ is called the valuation ex-dividend price of $A^{k_2}$, denoted by $\hat{S}_t(x,\varphi,A^{k_2},C^{k_3})$.
\end{definition}

The following assumptions are crucial for the next steps:
\begin{itemize}
\item[(i)] \label{PricingAssump1} The assumptions of Proposition \ref{prop:aoabasemkt} are met. This means in particular that the exists a probability measure $\QQ^e$ on $(\Omega,\cG_T)$, such that the processes \eqref{eq:firstMartingale} and \eqref{eq:secondMartingale}, i.e.
\begin{align*} 
\left(\int_{(0,t]}\left(\mathcal{X}^{e,k_1}_ud\left(\frac{S^{i,k_1}}{B^{i,k_1}}\right)_u+\frac{\mathcal{X}^{e,k_1}_u}{B^{i,k_1}_u}dD^{i,k_1}_u+d\left[\frac{S^{i,k_1}}{B^{i,k_1}},\mathcal{X}^{e,k_1}\right]_u\right)\right)_{0\leq t\leq T},
\end{align*}
\begin{align*}
\left(\frac{\mathcal{X}^{e,k_1}_tB^{k_1}_t}{B^e_t}\right)_{0\leq t\leq T},
\end{align*}
$i=1,\ldots,d_{k_1}$, $k_1=1\ldots,L$ are local martingales under the assumption that the repo constraint \eqref{eq:repoConstraint} is fulfilled.
\item[(ii)]\label{PricingAssump2} The collateral process $C^{k_3}$ is independent of the hedger's portfolio $\varphi$.
\end{itemize}

\begin{notation} \
\begin{itemize}
\item[•] In the sequel we will make use of the notation $\hat{A}^c$ or $\bar{A}^c$ to stress out the impact of the collateral in form of cash or respectively risky asset $S^{d_{k_3}+1,k_3}$, independent of its collateral convention and the contractual cash flows. In the cash collateral case we have that $\hat{A}^c \in \{\hat{F}^{s}+A^{e,k_2}, \hat{F}^h + A^{e,k_2}\}$ and if the hedger posts or receives risky asset collateral, $\bar{A}^c$ is given by $\bar{A}^c \in \{\bar{F}^s + A^{e,k_2}, \bar{F}^h + A^{e,k_2} \}$. To ensure that the integrals over the FX-processes $\mathcal{X}^{e,k_1}$ are well-defined, we assume those to be finite. The cash account $B^e$ remains an increasing process.
\item[•] Let $\QQ^e$ be a martingale measure for the discounted cumulative dividend price processes $\hat{S}^{i,cld,e,k_1}$ with $i \in \{1, \ldots d_{k_1}\}_{k_1 = 1, \ldots,L}\cup \{ d_{k_3}+1\}_{k_3 \in \{1, \ldots, L\}}$.
\item[•] For any $t \in [0,T]$, denote the ex-dividend price by $S_t(x,\varphi,A^{k_2},C^{k_3}):= S_t(A^{k_2},C^{k_3})$.
\end{itemize}
\end{notation}

\textcolor{black}{Since we are not considering a particular model setup, we assume that the random variables considered in the sequel  are integrale and we use the notation  $\EE_t^{\QQ^e}(\cdot) := \EE_{\QQ^e}(\cdot \mid \cG_t)$ to indicate the conditional expectation of some integrable random variable for any $t \in [0,T]$}.

\begin{proposition}\label{prop:pricinCollExo}
Assume that (i)-(ii) hold and the collateralized contract $(A^{k_2},C^{k_3})$ can be replicated by an admissible trading strategy $(x,\varphi,A^{k_2},C^{k_3})$ on the interval $[0,T]$. If the stochastic integrals with respect to \eqref{eq:firstMartingale}, for the indices $i = 1, \ldots, d_{k_1}$ with $k_1 = 1, \ldots, L$ and for the risky asset collateral case $k_3 \in \{1, \ldots,L\}$ are $\QQ^e$-martingales, then its corresponding ex-dividend price process $S(x,\varphi,A^{k_2},C^{k_3})$ is independent of $(x,\varphi)$ and equals 
\begin{align}
\label{eq:pricingcashcoll}
S_t(A,C)= -B_t^e \EE_t^{\QQ^e} \left( \int_{(t,T]} (B_u^e)^{-1} d\hat{A}_u^{c} \right),
\end{align}
for the cash collateral case and respectively 
\begin{align}
\label{eq:pricingriskycoll}
S_t(A,C)= -B_t^e \EE_t^{\QQ^e} \left( \int_{(t,T]} (B_u^e)^{-1} d\bar{A}_u^{c} \right),
\end{align}
for the risky asset collateral case.
\end{proposition}

\begin{proof}
We will start with the case, where the hedger posts or receives collateral in form of shares of risky asset $S ^{d_{k_3}+1,k_3}$ since the cash collateral case will follow immediately. Assume that there exists an admissible trading strategy $(x,\varphi,A^{k_2},C^{k_3})$ replicating the collateralized contract $(A^{k_2},C^{k_3})$ on the interval $(t,T]$. With \eqref{eq:definitionKi}, i.e. $dK_t^{i,e,k_1} = B_t^{e,k_1} d\hat{S}_t^{i,cld,e,k_1}$ for $i,k_1,t$ fulfilling the usual conditions including the risky asset collateral $S_t^{d_{k_3}+1,k_3}$, the dynamics of the discounted wealth process $\tilde{V}(x,\varphi,A^{k_2},C^{k_3}):= \tilde{V}(\varphi)$ are given by
\begin{align}
\begin{aligned}
d\tilde{V}_t(\varphi) &= d((B_t^e)^{-1}V_t(\varphi)) = V_t(\varphi) d(B_t^e)^{-1} + (B_t^e)^{-1} d\bar{A}_t^{c}\\
&= \sum_{k_1=1}^L \sum_{i=1}^{d_{k_1}} \xi_t^{i,k_1} \left(\tilde{B}_t^{i,k_1} d\hat{S}_t^{i,cld,e,k_1} - S_t^{i,k_1} d\mathcal{X}_t^{e,k_1} \right) + \sum_{k_1=1, k_1 \neq e}^L \psi_t^{0,k_1} B_t^e d\left(\frac{B^{k_1}\mathcal{X}^{e,k_1}}{B^e}  \right)_t \\
&+ \xi_t^{d_{k_3}+1,k_3} \left( \tilde{B}_t^{d_{k_3}+1,k_3}d\hat{S}_t^{d_{k_3}+1, cld,e ,k_3} - S_t^{d_{k_3}+1,k_3}d\mathcal{X}_t^{e,k_3} \right) \\
&+ (B_t^e)^{-1}d\bar{A}_t^{c},
\end{aligned}
\end{align}
by using Ito's formula and the dynamics of the hedger's wealth process \eqref{eq:dVriskycoll}. By definition, this process is self financing and fulfils the repo constraint \eqref{eq:repoConstraint}. We fix some $t \in [0,T)$. By assumption, there exists a replicating trading strategy 
\begin{align*}
(V_t^0(x) + p_t^e, \varphi, A^{k_2} - A_t^{k_2},C^{k_3}),
\end{align*}
in the sense of Definition \ref{def:replicating} fulfilling $V_T(V_t^0(x) + p_t^e,\varphi, A^{k_2} - A_t^{k_2},C^{k_3})=V_T^0(x) $ and hence
\begin{align*}
V_T(x,\varphi,A^{k_2},C^{k_3}) = V_T(x + (B_T^e)^{-1} p_T^e) = (x + (B_T^e)^{-1} p_T^e)B_T^e,
\end{align*}
leading to 
\begin{align*}
- (B_t^e)^{-1} p_t^e &= (B_T^{e})^{-1}V_T^e(x, \varphi, A^{k_2}, C^{k_3}) - (B_t^{e})^{-1}V_t^e(x, \varphi, A^{k_2}, C^{k_3})\\
&= \sum_{k_1=1}^L \sum_{i=1}^{d_{k_1}} \int_{(t,T]} \xi_u^{i,k_1} \left( \tilde{B}_u^{i,k_1} d\hat{S}_u^{i,cld,e,k_1}  - S_u^{i,k_1}d\mathcal{X}_u^{e,k_1}\right)\\
& +  \sum_{k_1=1, k_1 \neq e}^L  \int_{(t,T]} \psi_u^{0,k_1} B_u^e d\left(\frac{B^{k_1}\mathcal{X}^{e,k_1}}{B^e}  \right)_u\\
&+ \int_{(t,T]} \xi_u^{d_{k_3}+1,k_3} \left( \tilde{B}_u^{d_{k_3}+1,k_3} d\hat{S}_u^{i,cld,e,k_3} - S_u^{d_{k_1}+1,k_3} d\mathcal{X}_u^{e,k_3} \right)+ \int_{(t,T]} (B_u^e)^{-1} d\bar{A}_u^c
\end{align*}
by using that 
\begin{align}
\begin{aligned}
\tilde{B}_t^{i,k_1} d\hat{S}_t^{i,cld,e,k_1} - S_t^{i,k_1} d\mathcal{X}_t^{e,k_1} = \mathcal{X}_t^{e,k_1} d\left(\frac{S^{i,k_1}}{B^{i,k_1}}  \right)_t + \frac{\mathcal{X}_t^{e,k_1}}{B_t^{i,k_1}}dD_t^{i,k_1} +  d\left[\frac{S^{i,k_1}}{B^{i,k_1}},\mathcal{X}^{e,k_1} \right]_t ,
\end{aligned}
\end{align}
for all indices mentioned above as a direct consequence out of equation \eqref{eq:vNetLocalMartingale} and \eqref{eq:dynamicsK}.
Since the integrals with respect to \eqref{eq:firstMartingale} and \eqref{eq:secondMartingale} are true $\QQ^e$-martingales for $i=1, \ldots, d_{k_1}$, $k_1 = 1, \ldots, L$ and $k_3 \in \{1, \ldots, L \}$ , the ex-dividend price of the collateralized contract $(A^{k_2}, C^{k_3})$ can be computed via
\begin{align*}
S_t (&A^{k_2},C^{k_3})\\
& = - B_t^e \EE_{\QQ^e} \left( \sum_{k_1=1}^L \sum_{i=1}^{d_{k_1}} \int_{(0,T]} \xi_u^{i,k_1} \left( \mathcal{X}_u^{e,k_1} d\left(\frac{S^{i,k_1}}{B^{i,k_1}}  \right)_u + \frac{\mathcal{X}_u^{e,k_1}}{B_u^{i,k_1}}dD_u^{i,k_1} +  d\left[\frac{S^{i,k_1}}{B^{i,k_1}},\mathcal{X}^{e,k_1} \right]_u \right) \right. \\
&\left. - \sum_{k_1=1}^L \sum_{i=1}^{d_{k_1}} \int_{(0,t]} \xi_u^{i,k_1}  \left( \mathcal{X}_u^{e,k_1} d\left(\frac{S^{i,k_1}}{B^{i,k_1}}  \right)_u + \frac{\mathcal{X}_u^{e,k_1}}{B_u^{i,k_1}}dD_u^{i,k_1} +  d\left[\frac{S^{i,k_1}}{B^{i,k_1}},\mathcal{X}^{e,k_1} \right]_u \right) \mid \cG_t \right) \\
& - B_t^e \EE_{\QQ^e} \left(\int_{(0,T]} \xi_u^{d_{k_3}+1,k_3} \left(\mathcal{X}_u^{e,k_3} d\left(\frac{S^{d_{k_3}+1,k_3}}{B^{d_{k_3}+1,k_3}}  \right)_u\right.\right.\\
&\left.\left.\quad\quad\quad\quad + \frac{\mathcal{X}_u^{e,k_3}}{B_u^{d_{k_3}+1,k_3}}dD_u^{d_{k_3}+1,k_3} +  d\left[\frac{S^{d_{k_3}+1,k_3}}{B^{d_{k_3}+1,k_3}},\mathcal{X}^{e,k_3} \right]_u \right) \right.\\
&   - \int_{(0,t]} \xi_u^{d_{k_3}+1,k_3} \left(\mathcal{X}_u^{e,k_3} d\left(\frac{S^{d_{k_3}+1,k_3}}{B^{d_{k_3}+1,k_3}}  \right)_u\right.\\
&\quad\quad\quad\quad\left.\left. + \frac{\mathcal{X}_u^{e,k_3}}{B_u^{d_{k_3}+1,k_3}}dD_u^{d_{k_3}+1,k_3} +  d\left[\frac{S^{d_{k_3}+1,k_3}}{B^{d_{k_3}+1,k_3}},\mathcal{X}^{e,k_3} \right]_u \right) \mid \cG_t \right)\\
& - B_t^e \EE_{\QQ^e} \left( \int_{(0,T]} (B_u^e)^{-1} d\bar{A}_u^c  - \int_{(0,t]} (B_u^e)^{-1} d\bar{A}_u^c \mid \cG_t \right)\\
=& -B_t^e \EE_{\QQ^e} \left( \int_{(t,T]} (B_u^e)^{-1} d\bar{A}_u^c \mid \cG_t \right) = -B_t^e \EE_t^{\QQ^e} \left( \int_{(t,T]} (B_u^e)^{-1} d\bar{A}_u^{c} \right),
\end{align*}
for any $t \in [0,T]$ by using the martingale and measurability properties and is independent of $(x, \varphi)$. By following the same steps for the cash collateral case, \eqref{eq:pricingcashcoll} follows immediately.
\end{proof}
\color{black}

\begin{remark}\label{eq:firstPricingWithColl}
Note that in case of absolute continuity of all repo accounts, the ex-dividend price process $S_t(A^{k_2},C^{k_3})$ for any $t \in [0,T]$ is given as follows:
\begin{itemize}
\item[1)] \underline{Cash collateral under segregation:}
By using equation \eqref{eq:hatFsegregation}, we write
\begin{align}
\begin{aligned}
S_t(A^{k_2},C^{k_3})=& -B_t^e \EE_t^{\QQ^e} \left[ \int_{(t,T]} (B_u^e)^{-1} dA_u^{e,k_2} \right] - B_t^e \EE_t^{\QQ^e} \left[ \int_{(t,T]} (B_u^e)^{-1} d\hat{F}_u^s \right]\\
=& -B_t^e \EE_t^{\QQ^e} \left[ \int_{(t,T]} (B_u^e)^{-1} dA_u^{e,k_2} \right]\\
& -B_t^e \EE_t^{\QQ^e} \left[ \int_{(t,T]} \left[ \left( r_u^{d_{k_3}+2,k_2,s} (C_u^{k_3})^+  - r_u^{d_{k_3}+3,k_3}(C_u^{k_3})^- \right.\right.\right.\\
& \left.\left.\left.+ r_u^{c,k_3,l}(C_u^{k_3})^- - r_u^{c,k_3,b}(C_u^{k_3})^+ \right]\mathcal{X}_u^{e,k_3}du - C_u^{k_3} d\mathcal{X}_u^{e,k_3} \right) \frac{1}{B_u^e}\right]\\
=& -B_t^e \EE_t^{\QQ^e} \left[ \int_{(t,T]} (B_u^e)^{-1} dA_u^{e,k_2} \right]\\
& - B_t^e \EE_t^{\QQ^e} \left[ \int_{(t,T]} \frac{1}{B_u^e} \left[\left( r_u^{d_{k_3}+2,k_3,s} - r_u^{c,k_3,b}\right) (C_u^{k_3})^+ \right. \right.\\
& \left. \left. - \left(r_u^{d_{k_3}+3,k_3} - r_u^{c,k_3,l}\right) (C_u^{k_3})^-  \right]\mathcal{X}_u^{e,k_3}du - \int_{(t,T]} \frac{C_u^{k_3}}{B_u^e} d\mathcal{X}_u^{e,k_3} \right].
\end{aligned}
\end{align}
\item[2)] \underline{Cash collateral under rehypothecation:} With similar calculations by using equation \eqref{eq:hatFrehypothecation}, we get
\begin{align}
\begin{aligned}
S_t(A^{k_2},C^{k_3}) =& -B_t^e \EE_t^{\QQ^e} \left[ \int_{(t,T]} (B_u^e)^{-1} dA_u^{e,k_2} \right]\\
& -B_t^e \EE_t^{\QQ^e} \left[ \int_{(t,T]} \frac{1}{B_u^e} \left[ \left( r_u^e - r_u^{c,k_3,b} \right)(C_u^{k_3})^+  \right. \right.\\
&\left. \left. - \left(r_u^{d_{k_3}+3,k_3} - r_u^{c,k_3,l}  \right)(C_u^{k_3})^- \right] \mathcal{X}_u^{e,k_3} du - \int_{(t,T]} \frac{C_u^{k_3}}{B_u^e} d\mathcal{X}_u^{e,k_3} \right].
\end{aligned}
\end{align}
\item[3)] \underline{Risky asset collateral under segregation:} By using equation \eqref{eq:barFriskyassetcoll}, we derive
\begin{align}
\begin{aligned}
S_t(A^{k_2},C^{k_3}) = &-B_t^e \EE_t^{\QQ^e} \left[ \int_{(t,T]} (B_u^e)^{-1} dA_u^{e,k_2} \right]\\
& -B_t^e \EE_t^{\QQ^e} \left[ \int_{(t,T]} \frac{1}{B_u^e} \left[ \left(r_u^{d_{k_3}+2,k_3,s} - r_u^{c,k_3,b}  \right)(C_u^{k_3})^+ \right. \right.\\
& \left. \left. - \left(r_u^{d_{k_3}+1,k_3} - r_u^{c,k_3,l} \right) (C_u^{k_3})^- \right] \mathcal{X}_u^{e,k_3}du - \int_{(t,T]} \frac{C_u^{k_3}}{B_u^e} d\mathcal{X}_u^{e,k_3} \right].
\end{aligned}
\end{align}
\item[4)] \underline{Risky asset collateral under rehypothecation:} Analogously, replacing index $s$ by $h$, we receive
\begin{align}
\begin{aligned}
S_t(A^{k_2},C^{k_3}) = &-B_t^e \EE_t^{\QQ^e} \left[ \int_{(t,T]} (B_u^e)^{-1} dA_u^{e,k_2} \right]\\
& -B_t^e \EE_t^{\QQ^e} \left[ \int_{(t,T]} \frac{1}{B_u^e} \left[ \left(r_u^{d_{k_3}+2,k_3,h} - r_u^{c,k_3,b}  \right)(C_u^{k_3})^+ \right. \right.\\
& \left. \left. - \left(r_u^{d_{k_3}+1,k_3} - r_u^{c,k_3,l} \right) (C_u^{k_3})^- \right] \mathcal{X}_u^{e,k_3}du - \int_{(t,T]} \frac{C_u^{k_3}}{B_u^e} d\mathcal{X}_u^{e,k_3} \right].
\end{aligned}
\end{align}
\end{itemize}
\end{remark}

\section{Diffusion models}\label{sec:DiffusionModels}

The aim of the present section is to provide concrete examples concerning the valuation of cross currency products. 
The diffusion model we present can be thought of as a footprint to construct cross currency simulation models for the computation of various valuation adjustments known in the literature under the acronym of xVA.

\textcolor{black}{We fix a filtered probability space $(\Omega, \mathcal{G}, \mathbb{G}, \mathbb{P})$, where the filtration $\mathbb{G}=(\mathcal{G})_{t \in[0, T]}$ satisfies the usual conditions. We assume that $\mathcal{G}_0$ is trivial.} We will assume that all cash accounts are absolutely continuous with respect to the Lebesgue measure, so that they can be written in the form $dB_t^\cdot=r^\cdot_t B^\cdot_tdt$ for some $\GG$-adapted RCLL bounded processes $r^\cdot$. In each currency area $k_1=1,\ldots,L$ we postulate the existence of $d_{k_1}$ traded risky assets $S^{i,k_1}$. For the collateral currency $k_3$ we also postulate the existence of the traded assets $S^{d_{k_3}+1,k_1}$ and $S^{d_{k_3}+2,k_1}$. Finally, we also assume that the repo constraint \eqref{eq:repoConstraint} is satisfied.

\subsection{Model Dynamics and martingale measure}

We construct the model and the domestic martingale measure $\mathbb{Q}^e$. In line with the single currency model of \cite{BieRut15} we postulate the following dynamics for each asset $S^{i,k_1}$ under the physical measure $\mathbb{P}$. \textcolor{black}{Risky assets evolve according to SDEs, defined on $(\Omega, \mathcal{G}, \mathbb{G}, \mathbb{P})$, of the form}
\begin{align}
dS^{i,k_1}_t=S^{i,k_1}_t\left(\mu^{S^{i,k_1}}_tdt+\sigma^{S^{i,k_1}}_tdW^{S^{i,k_1},\mathbb{P}}_t\right),
\end{align}
for $k_1=1,\ldots ,L$, $k_3\in\{1,\ldots ,L\}$,  $i=1,\ldots ,d_{k_1}$ for the hedging assets and $d_{k_3}+1,d_{k_3}+2$ for the collateral assets. The \textcolor{black}{$\GG$-adapted} drift functions $\mu^{S^{i,k_1}}$ are bounded while the \textcolor{black}{$\GG$-adapted} volatility functions $\sigma^{S^{i,k_1}}$ are strictly positive and bounded. The \textcolor{black}{$\mathcal{G}_t$-}Brownian motions $W^{S^{i,k_1},\mathbb{P}}$ are correlated via
\begin{align*}
d\left\langle W^{S^{i, k_{1}}, \mathbb{P}}, W^{S^{i^\prime, k^\prime_1}, \mathbb{P}}\right\rangle_t = \rho_t^{S^{i,k_1},S^{i^\prime,k^\prime_1}}dt,
\end{align*}
for $-1\leq \rho^{S^{i,k_1},S^{i^\prime,k^\prime_1}}\leq 1$.
The \textcolor{black}{$\GG$-adapted} dividend processes of the risky assets are given by $D^{i,k_1}_t=\int_0^t \kappa^{i,k_1}_uS^{i,k_1}_udu$, where the bounded processes $\kappa^{i,k_1}$ represent dividend yields. We also assume that exchange rates evolve according to \textcolor{black}{to SDEs, defined on $(\Omega, \mathcal{G}, \mathbb{G}, \mathbb{P})$, of the form}
\begin{align}
d\mathcal{X}^{e,k_1}_t=\mathcal{X}^{e,k_1}_t\left(\mu^{\mathcal{X}^{e,k_1}}_tdt+\sigma^{\mathcal{X}^{e,k_1}}_tdW^{\mathcal{X}^{e,k_1},\mathbb{P}}_t\right),
\end{align}
with analogous assumptions on drifts and volatilities. \textcolor{black}{We allow for correlations among exchange rates, via
\begin{align*}
&d\left\langle W^{\mathcal{X}^{e,k_1}, \mathbb{P}}, W^{\mathcal{X}^{e,k^\prime_1}, \mathbb{P}}\right\rangle_t = \rho_t^{\mathcal{X}^{e,k_1},\mathcal{X}^{e,k^\prime_1}}dt,\\
&d\left\langle W^{\mathcal{X}^{e,k_1}, \mathbb{P}}, W^{S^{i,k^\prime_1}, \mathbb{P}}\right\rangle_t = \rho_t^{\mathcal{X}^{e,k_1},S^{i,k^\prime_1}}dt,
\end{align*}
for $-1\leq\rho^{\mathcal{X}^{e,k_1},\mathcal{X}^{e,k^\prime_1}}\leq 1$ and  $-1\leq\rho^{\mathcal{X}^{e,k_1},S^{i,k^\prime_1}}\leq 1$}, for $k_1,k_1^\prime=1,\ldots, L$ and $i=1,\ldots,d_{k_1}$. The correlation coefficient functions are such that the resulting correlation matrix is positive semi-definite.

The following generalizes Lemma 5.2 in \cite{BieRut15}.

\begin{lemma}\label{lem:dynamicsTradedAssets}
Under the measure $\mathbb{Q}^e$ the following holds.
\begin{enumerate}
\item The dynamics of domestic assets $S^{i,e}$ are of the form
\begin{align}
dS^{i,e}_t=S^{i,e}_t\left((r^{i,k_1}_t-\kappa^{i,k_1}_t)dt+\sigma^{S^{i,e}}_tdW^{S^{i,e},\mathbb{Q}^e}_t\right).
\end{align}
Equivalently
\begin{align*}
d\hat{S}^{i,cld,e}_t=\hat{S}^{i,cld,e}_t\sigma^{S^{i,e}}_tdW^{S^{i,e},\mathbb{Q}^e}_t
\end{align*}
and
\begin{align*}
d{K}^{i,e,e}_t=dS^{i,e}_t-r^{i,e}_tS^{i,e}_tdt+\kappa^{i,e}_tS^{i,e}_tdt=S^{i,e}_t\sigma^{S^{i,e}}_tdW^{S^{i,e},\mathbb{Q}^e}_t
\end{align*}
are local martingales under $\mathbb{Q}^e$.
\item The dynamics of foreign assets $S^{i,k_1}$ are of the form
\begin{align}
dS^{i,k_1}_t=S^{i,k_1}_t\left((r^{i,k_1}_t-\kappa^{i,k_1}_t-\rho^{S^{i,k_1},\mathcal{X}^{e,k_1}}_t\sigma^{S^{i,k_1}}_t\sigma^{\mathcal{X}^{e,k_1}}_t)dt+\sigma^{S^{i,k_1}}_tdW^{S^{i,k_1},\QQ^e}_t\right)
\end{align}
and the processes
\begin{align*}
d K_{t}^{i, e, k_{1}}-S_{t}^{i, k_{1}} d \mathcal{X}_{t}^{e, k_{1}}=\mathcal{X}^{e,k_1}_tS^{i,k_1}_t\sigma^{S^{i,k_1}}_tdW^{S^{i,k_1},\QQ^e}_t
\end{align*}
are local martingales under $\mathbb{Q}^e$.
\item The dynamics of all exchange rates are of the form
\begin{align}
d\mathcal{X}^{e,k_1}_t=\mathcal{X}^{e,k_1}_t\left((r^e_t-r^{k_1}_t)dt+\sigma^{\mathcal{X}^{e,k_1}}_tdW^{\mathcal{X}^{e,k_1},\QQ^e}_t\right)
\end{align}
and the processes
\begin{align*}
d\left(\frac{\mathcal{X}^{e,k_1}_tB^{k_1}_t}{B^e_t}\right)=\frac{\mathcal{X}^{e,k_1}_tB^{k_1}_t}{B^e_t}\sigma^{\mathcal{X}^{e,k_1}}_tdW^{\mathcal{X}^{e,k_1},\QQ^e}_t
\end{align*}
are local martingales under $\mathbb{Q}^e$.
\end{enumerate}
\end{lemma}

\begin{proof}
The statement on the domestic assets corresponds to that of Lemma 5.2 in \cite{BieRut15} and thus the proof is omitted. Let us concentrate on the foreign assets. Under $\mathbb{Q}^e$, the process \eqref{eq:firstMartingale} is a local martingale. The quadratic covariation between the repo-discounted  asset price $\frac{S^{i,k_1}}{B^{i,k_1}}$ and the exchange rate is 
\begin{align*}
\left[\frac{S^{i,k_1}}{B^{i,k_1}},\mathcal{X}^{e,k_1}\right]_t=\left\langle\frac{S^{i,k_1}}{B^{i,k_1}},\mathcal{X}^{e,k_1}\right\rangle_t=\int_{(0,t]}\rho^{S^{i,k_1},\mathcal{X}^{e,k_1}}_u\sigma^{S^{i,k_1}}_u\sigma^{\mathcal{X}^{e,k_1}}_u\frac{S^{i,k_1}_u\mathcal{X}^{e,k_1}_u}{B^{i,k_1}_u}du.
\end{align*} 
We can write in explicit form
\begin{align*}
&d K_{t}^{i, e, k_{1}}-S_{t}^{i, k_{1}} d \mathcal{X}_{t}^{e, k_{1}}\\
&\quad=B^{i,k_1}_t\left(\mathcal{X}^{e,k_1}_td\left(\frac{S^{i,k_1}}{B^{i,k_1}}\right)_t+\frac{\mathcal{X}^{e,k_1}_t}{B^{i,k_1}_u}dD^{i,k_1}_t+d\left[\frac{S^{i,k_1}}{B^{i,k_1}},\mathcal{X}^{e,k_1}\right]_t\right)\\
&\quad=B^{i,k_1}_t\left(\frac{\mathcal{X}^{e,k_1}_t}{B^{i,k_1}_t}S^{i,k_1}_t\left(\mu^{S^{i,k_1}}_tdt+\sigma^{S^{i,k_1}}_tdW^{S^{i,k_1},\mathbb{P}}_t\right)-r^{i,k_1}_t\frac{S^{i,k_1}_t\mathcal{X}^{e,k_1}_t}{B^{i,k_1}_t}dt+\kappa^{i,k_1}_t\frac{S^{i,k_1}_t\mathcal{X}^{e,k_1}_t}{B^{i,k_1}_t}dt\right.\\
&\quad\quad\left.+\rho^{S^{i,k_1},\mathcal{X}^{e,k_1}}_t\sigma^{S^{i,k_1}}_t\sigma^{\mathcal{X}^{e,k_1}}_t\frac{S^{i,k_1}_t\mathcal{X}^{e,k_1}_t}{B^{i,k_1}_t}dt\right)\\
&\quad=\mathcal{X}^{e,k_1}_tS^{i,k_1}_t\left((\mu^{S^{i,k_1}}_t-r^{i,k_1}_t+\kappa^{i,k_1}_t+\rho^{S^{i,k_1},\mathcal{X}^{e,k_1}}_t\sigma^{S^{i,k_1}}_t\sigma^{\mathcal{X}^{e,k_1}}_t)dt+\sigma^{S^{i,k_1}}_tdW^{S^{i,k_1},\mathbb{P}}_t\right).
\end{align*}
If the process
\begin{align*}
dW^{S^{i,k_1},\QQ^e}_t:=dW^{S^{i,k_1},\PP}_t+\frac{1}{\sigma^{S^{i,k_1}}_t}\left(\mu^{S^{i,k_1}}_t-r^{i,k_1}_t+\kappa^{i,k_1}_t+\rho^{S^{i,k_1},\mathcal{X}^{e,k_1}}_t\sigma^{S^{i,k_1}}_t\sigma^{\mathcal{X}^{e,k_1}}_t\right)dt
\end{align*}
is a Brownian motion under $\QQ^e$ then
\begin{align*}
d K_{t}^{i, e, k_{1}}-S_{t}^{i, k_{1}} d \mathcal{X}_{t}^{e, k_{1}}=\mathcal{X}^{e,k_1}_tS^{i,k_1}_t\sigma^{S^{i,k_1}}_tdW^{S^{i,k_1},\QQ^e}_t
\end{align*}
is a local martingale under $\QQ^e$. Finally, for the dynamics of the asset we obtain
\begin{align*}
dS^{i,k_1}_t&=S^{i,k_1}_t\left(\mu^{S^{i,k_1}}_tdt+\sigma^{S^{i,k_1}}_tdW^{S^{i,k_1},\mathbb{P}}_t\right)\\
&=S^{i,k_1}_t\left((r^{i,k_1}_t-\kappa^{i,k_1}_t-\rho^{S^{i,k_1},\mathcal{X}^{e,k_1}}_t\sigma^{S^{i,k_1}}_t\sigma^{\mathcal{X}^{e,k_1}}_t)dt+\sigma^{S^{i,k_1}}_tdW^{S^{i,k_1},\QQ^e}_t\right),
\end{align*}
which completes the proof of the second statement.

For the exchange rates we proceed analogously. Under $\QQ^e$ we require that the process \eqref{eq:secondMartingale} is a local martingale. The computation is straightforward. We have
\begin{align*}
d\left(\frac{\mathcal{X}^{e,k_1}_tB^{k_1}_t}{B^e_t}\right)=\frac{\mathcal{X}^{e,k_1}_tB^{k_1}_t}{B^e_t}\left((\mu^{\mathcal{X}^{e,k_1}}_t+r^{k_1}_t-r^e_t)dt+\sigma^{\mathcal{X}^{e,k_1}}_tdW^{\mathcal{X}^{e,k_1},\mathbb{P}}_t\right).
\end{align*}
If the process
\begin{align*}
dW^{\mathcal{X}^{e,k_1},\QQ^e}_t=dW^{\mathcal{X}^{e,k_1},\PP}_t+\frac{1}{\sigma^{\mathcal{X}^{e,k_1}}_t}\left(\mu^{\mathcal{X}^{e,k_1}}_t+r^{k_1}_t-r^e_t\right)dt
\end{align*}
is a Brownian motion under $\QQ^e$ then we obtain a local martingale and the resulting dynamics of the exchange rates are given by
\begin{align*}
d\mathcal{X}^{e,k_1}_t=\mathcal{X}^{e,k_1}_t\left((r^e_t-r^{k_1}_t)dt+\sigma^{\mathcal{X}^{e,k_1}}_tdW^{\mathcal{X}^{e,k_1},\QQ^e}_t\right),
\end{align*}
which completes the proof.
\end{proof}

The dynamics we obtained above are of independent interest: they provide a sound framework for the construction of hybrid models for a multitude of risky assets in a multi currency setting. Such models can be used for the Monte Carlo simulation of risk factors that affect a portfolio of contingent claims. Such high-dimensional hybrid models for a multitude of risk factors constitute the market standard for the computation of valuation adjustments (xVA) for a whole portfolio of claims between the hedger and the counterparty. A by product of our valuation framework is then a sound derivation of multi-currency hybrid models for the generation of exposure profiles for counterparty credit risk. Hybrid models for xVA are presented in \cite{sokol2014}, \cite{green2015}, \cite{listag2015}.

The basic model above can be extended in multiple directions: our choice for the driving processes is rather simplicistic and mainly meant to provide an illustration of how one can construct a cross currency hybrid model in a multi curve framework. One natural stream of generalization is to consider more general driving processes. One possibility is to extend the market by introducing instruments which are by definition fully collateralized, i.e. natively collateralized assets such as OIS bonds and (textbook) FRAs as in \cite{Cuchiero2016} and \cite{Cuchiero2019}. The resulting model would allow for the joint evolution of interbank spreads, overnight rates, foreign exchange and risky assets. We leave such extensions to future research.

\subsection{Wealth dynamics with collateral}
We can now provide explicit expressions for the wealth dynamics under any collateralization scheme thanks to Lemma \ref{lem:dynamicsTradedAssets}. We assume again, as in Remark \ref{eq:firstPricingWithColl} that all cash accounts are absolutely continuous. In line with Section \ref{sec:pricingExogColl}, we assume for the moment that the collateral $C^{k_3}$ is exogenously given.

\subsubsection{Cash Collateral under segregation} In Proposition \ref{prop:cashSegregation} we have that \eqref{eq:VcashSegregation} takes now the form  
\begin{align}
\begin{aligned}
dV_t(\varphi) &= V_t(\varphi)r_t^edt + \sum_{k_1=1}^L  \sum_{i=1}^{d_{k_1}} \xi_t^{i,k_1}\mathcal{X}^{e,k_1}_tS^{i,k_1}_t\sigma^{S^{i,k_1}}_tdW^{S^{i,k_1},\QQ^e}_t\\
&+\sum_{k_1=1,k_1 \neq e}^L\psi_t^{0,k_1}B^{k_1}_t\mathcal{X}^{e,k_1}_t\sigma^{\mathcal{X}^{e,k_1}}_tdW^{\mathcal{X}^{e,k_1},\QQ^e}_t  + dA_t^{e,k_2} +d\hat{F}_t^s,
\end{aligned}
\end{align}
with
\begin{align*}
\hat{F}_t^s&=\int_0^t\left( r_u^{d_{k_3}+2,k_3,s} - r_u^{c,k_3,b}\right) (C_u^{k_3})^+\mathcal{X}_u^{e,k_3}du\\
&\quad -\int_0^t \left(r_u^{d_{k_3}+3,k_3} - r_u^{c,k_3,l}\right) (C_u^{k_3})^-\mathcal{X}_u^{e,k_3}du - \int_{(0,t]} C_u^{k_3} d\mathcal{X}_u^{e,k_3}.
\end{align*}

\subsubsection{Cash Collateral under rehypothecation}
In Proposition \ref{prop:cashRehypothecation} we have that \eqref{eq:VcashRehypothecation} takes now the form
\begin{align}
\label{eq:ccrehypDiff}
\begin{aligned}
dV_t(\varphi) &= V_t(\varphi)r_t^edt + \sum_{k_1=1}^L  \sum_{i=1}^{d_{k_1}} \xi_t^{i,k_1}\mathcal{X}^{e,k_1}_tS^{i,k_1}_t\sigma^{S^{i,k_1}}_tdW^{S^{i,k_1},\QQ^e}_t\\
&+\sum_{k_1=1,k_1 \neq e}^L\psi_t^{0,k_1}B^{k_1}_t\mathcal{X}^{e,k_1}_t\sigma^{\mathcal{X}^{e,k_1}}_tdW^{\mathcal{X}^{e,k_1},\QQ^e}_t  + dA_t^{e,k_2} +d\hat{F}_t^h,
\end{aligned}
\end{align}
with
\begin{align*}
\hat{F}_t^h&=\int_0^t\left( r_u^e - r_u^{c,k_3,b} \right) (C_u^{k_3})^+\mathcal{X}_u^{e,k_3}du\\
&\quad -\int_0^t\left(r_u^{d_{k_3}+3,k_3} - r_u^{c,k_3,l}  \right) (C_u^{k_3})^-\mathcal{X}_u^{e,k_3}du - \int_{(0,t]} C_u^{k_3} d\mathcal{X}_u^{e,k_3}.
\end{align*}
We observe that, for the case of cash collateral, the difference between segregation and rehypothecation is reflected only by the presence of $r^{d_{k_3}+2,k_3,s}$ and $r^e$ respectively.

\subsubsection{Risky asset collateral}

Risky asset collateral was treated in Proposition \ref{prop:RiskyCollateral} both under segregation and rehypothecation. In the diffusive setting of the present section \eqref{eq:dVriskycoll} under rehypothecation takes now the form
\begin{align}
\begin{aligned}
dV_t(\varphi) &=V_t(\varphi)r_t^edt + \sum_{k_1=1}^L  \sum_{i=1}^{d_{k_1}} \xi_t^{i,k_1}\mathcal{X}^{e,k_1}_tS^{i,k_1}_t\sigma^{S^{i,k_1}}_tdW^{S^{i,k_1},\QQ^e}_t\\
&+\sum_{k_1=1,k_1 \neq e}^L\psi_t^{0,k_1}B^{k_1}_t\mathcal{X}^{e,k_1}_t\sigma^{\mathcal{X}^{e,k_1}}_tdW^{\mathcal{X}^{e,k_1},\QQ^e}_t \\
&+ (S_t^{d_{k_3}+1,k_3})^{-1}(C_t^{k_3})^{-}\mathcal{X}^{e,k_3}_tS^{d_{k_3}+1,k_3}_t\sigma^{S^{d_{k_3}+1,k_3}}_tdW^{S^{d_{k_3}+1,k_3},\QQ^e}_t\\\
&+ dA_t^{e,k_2} +d\bar{F}_t^h 
\end{aligned}
\end{align} 
with
\begin{align*}
\bar{F}_t^h&=\int_0^t\left(r_u^{d_{k_3}+2,k_3,h} - r_u^{c,k_3,b}  \right) (C_u^{k_3})^+\mathcal{X}_u^{e,k_3}du\\
&\quad -\int_0^t\left(r_u^{d_{k_3}+1,k_3} - r_u^{c,k_3,l} \right) (C_u^{k_3})^-\mathcal{X}_u^{e,k_3}du - \int_{(0,t]} C_u^{k_3} d\mathcal{X}_u^{e,k_3}.
\end{align*}
The case of segregation is obtained by simply replacing $r^{d_{k_3}+2,k_3,h}$ with $r^{d_{k_3}+2,k_3,s}$.

\subsection{Pricing with exogenous collateral}
We specialize the findings of Proposition \ref{prop:pricinCollExo} to the diffusive setting of the present section. In line with \cite{BieRut15} we assume that the process $A^{e,k_2}$ is adapted to the filtration $\FF^{S,\mathcal{X}}$, generated by all risky assets and all exchange rates. $A^{c,k_2}$ is a shorthand for the processes employed in Proposition \ref{prop:pricinCollExo}. In the following we assume that all conditional expectations considered in the sequel are well defined for all $t\in[0,T]$.

\begin{proposition}
In the diffusion model, a collateralized contract $(A^{k_2},C^{k_3})$ with predetermined collateral process $C^{k_3}$ can be replicated by an admissible trading strategy. The ex-dividend price $S(A,C)$ satisfies, for very $t\in[0,T]$
\begin{align*}
S_t(A^{k_2},C^{k_3})=-B^e_t\EE_t^{\QQ^e} \left( \int_{(t,T]} (B_u^e)^{-1} d\hat{A}_u^{c} \right).
\end{align*}
\end{proposition}

\begin{proof}
The present result corresponds to \cite{BieRut15} Proposition 5.3.
\end{proof}

At this point, we would like to show that the formulas we developed allow us to link the general framework of the present paper with the findings of \cite{mopa17}, \cite{fushita09}, \cite{fushita10b}, \cite{fushita10c}, \cite{fushita10}. Let us recall that process $A^{e,k_2}$ satisfies $dA^{e,k_2}=\mathcal{X}^{e,k_2}_tdA^{k_2}_t$.

\begin{corollary}\label{cor:pricingExoCollExplicit}
In the diffusion model, we have the following pricing formulas for a collateralized contract $(A^{k_2},C^{k_3})$ with predetermined collateral process $C^{k_3}$.
\begin{itemize}
\item[1)] \underline{Cash collateral under segregation:}
\begin{align}
\begin{aligned}
S_t&(A^{k_2},C^{k_3})= -B_t^e \EE_t^{\QQ^e} \left[ \int_{(t,T]} \frac{\mathcal{X}^{e,k_2}_u}{B_u^e} dA_u^{k_2} \right]\\
& - B_t^e \EE_t^{\QQ^e} \left[ \int_{(t,T]} \left[\left( r_u^{d_{k_3}+2,k_3,s} - r_u^{c,k_3,b}\right) (C_u^{k_3})^+ \right. \right.\\
& \left. \left. - \left(r_u^{d_{k_3}+3,k_3} - r_u^{c,k_3,l}\right) (C_u^{k_3})^-  \right]\frac{\mathcal{X}_u^{e,k_3}}{B_u^e}du - \int_{(t,T]} \frac{C_u^{k_3}\mathcal{X}_u^{e,k_3}}{B_u^e}(r^e_u-r^{k_3}_u)du  \right].
\end{aligned}
\end{align}
\item[2)] \underline{Cash collateral under rehypothecation:}
\begin{align}
\label{eq:priceCashRehyp}
\begin{aligned}
S_t&(A^{k_2},C^{k_3}) = -B_t^e \EE_t^{\QQ^e} \left[ \int_{(t,T]} \frac{\mathcal{X}^{e,k_2}_u}{B_u^e} dA_u^{k_2} \right]\\
& -B_t^e \EE_t^{\QQ^e} \left[ \int_{(t,T]}\left[ \left( r_u^e - r_u^{c,k_3,b} \right)(C_u^{k_3})^+  \right. \right.\\
&\left. \left. - \left(r_u^{d_{k_3}+3,k_3} - r_u^{c,k_3,l}  \right)(C_u^{k_3})^- \right] \frac{\mathcal{X}_u^{e,k_3}}{B_u^e} du - \int_{(t,T]} \frac{C_u^{k_3}\mathcal{X}_u^{e,k_3}}{B_u^e}(r^e_u-r^{k_3}_u)du\right].
\end{aligned}
\end{align}
\item[3)] \underline{Risky asset collateral under segregation:}
\begin{align}
\begin{aligned}
S_t&(A^{k_2},C^{k_3}) = -B_t^e \EE_t^{\QQ^e} \left[ \int_{(t,T]} \frac{\mathcal{X}^{e,k_2}_u}{B_u^e} dA_u^{k_2} \right]\\
& -B_t^e \EE_t^{\QQ^e} \left[ \int_{(t,T]}\left[ \left(r_u^{d_{k_3}+2,k_3,s} - r_u^{c,k_3,b}  \right)(C_u^{k_3})^+ \right. \right.\\
& \left. \left. - \left(r_u^{d_{k_3}+1,k_3} - r_u^{c,k_3,l} \right) (C_u^{k_3})^- \right] \frac{\mathcal{X}_u^{e,k_3}}{B_u^e}du - \int_{(t,T]} \frac{C_u^{k_3}\mathcal{X}_u^{e,k_3}}{B_u^e}(r^e_u-r^{k_3}_u)du \right].
\end{aligned}
\end{align}
\item[4)] \underline{Risky asset collateral under rehypothecation:}
\begin{align}
\begin{aligned}
S_t&(A^{k_2},C^{k_3}) = -B_t^e \EE_t^{\QQ^e} \left[ \int_{(t,T]} \frac{\mathcal{X}^{e,k_2}_u}{B_u^e} dA_u^{k_2}\right]\\
& -B_t^e \EE_t^{\QQ^e} \left[ \int_{(t,T]} \left[ \left(r_u^{d_{k_3}+2,k_3,h} - r_u^{c,k_3,b}  \right)(C_u^{k_3})^+ \right. \right.\\
& \left. \left. - \left(r_u^{d_{k_3}+1,k_3} - r_u^{c,k_3,l} \right) (C_u^{k_3})^- \right] \frac{\mathcal{X}_u^{e,k_3}}{B_u^e}du - \int_{(t,T]}\frac{C_u^{k_3}\mathcal{X}_u^{e,k_3}}{B_u^e}(r^e_u-r^{k_3}_u)du\right].
\end{aligned}
\end{align}
\end{itemize}
\end{corollary}

\begin{proof}
The proof directly follows from Remark \ref{eq:firstPricingWithColl} and by observing that
\begin{align*}
 \EE_t^{\QQ^e} \left[ \int_{(t,T]}\frac{C_u^{k_3}}{B_u^e}d\mathcal{X}_u^{e,k_3}\right]=  \EE_t^{\QQ^e} \left[ \int_{(t,T]}\frac{C_u^{k_3}\mathcal{X}_u^{e,k_3}}{B_u^e}(r^e_u-r^{k_3}_u)du\right].
\end{align*}
\end{proof}

The existing literature focuses on the case of cash collateral with rehypothecation, for example \cite{mopa17} in their Proposition 1 obtain the analogue of \eqref{eq:priceCashRehyp}. Also different borrowing and lending rates are not considered. Our Corollary \ref{cor:pricingExoCollExplicit} generalises most results available in the literature since we allow for different combinations of collateralization covenants. The distinctive feature of pricing formulas, when collateral can be posted in different currencies, lies in the further "correction" term which is proportional to the drift of the FX rate and the we could compute explicitly in the present diffusive setting. As a final illustration, let us stress that each of the four pricing formulas above nests the three following valuation formulas.

\begin{remark} In the case of cash collateral with rehypothecation we have the following special cases of \eqref{eq:priceCashRehyp}.
\begin{itemize}
\item[2.a)] \underline{Domestic cashflows collateralized in domestic currency}. This corresponds to the case $k_2 = k_3 = e$ and we obtain
\begin{align}
\begin{aligned}
S_t(A^{e},C^{e}) =& -B_t^e \EE_t^{\QQ^e} \left[ \int_{(t,T]} \frac{1}{B_u^e} dA_u^{e} \right] -B_t^e \EE_t^{\QQ^e} \left[ \int_{(t,T]}\left[ \left( r_u^e - r_u^{c,e,b} \right)(C_u^{e})^+  \right. \right.\\
&\left. \left. - \left(r_u^{d_{e}+3,e} - r_u^{c,e,l}  \right)(C_u^{e})^- \right] \frac{1}{B_u^e} du\right].
\end{aligned}
\end{align}
This is the case already treated both in \cite{pit10} \cite{BieRut15} among others.
\item[2.b)]\underline{Domestic cashflows collateralized in foreign currency}. This corresponds to the case $k_2=e$ and $k_3 \neq e$ and we obtain
\begin{align}
\begin{aligned}
S_t(A^{e},C^{k_3}) =& -B_t^e \EE_t^{\QQ^e} \left[ \int_{(t,T]} \frac{1}{B_u^e} dA_u^{e} \right] -B_t^e \EE_t^{\QQ^e} \left[ \int_{(t,T]}\left[ \left( r_u^e - r_u^{c,k_3,b} \right)(C_u^{k_3})^+  \right. \right.\\
&\left. \left. - \left(r_u^{d_{k_3}+3,k_3} - r_u^{c,k_3,l}  \right)(C_u^{k_3})^- \right] \frac{\mathcal{X}_u^{e,k_3}}{B_u^e} du - \int_{(t,T]} \frac{C_u^{k_3}\mathcal{X}_u^{e,k_3}}{B_u^e}(r^e_u-r^{k_3}_u)du\right].
\end{aligned}
\end{align}
\item[2.c)]\underline{Foreign cashflows collateralized in domestic currency}. This corresponds to the case $k_2\neq e$ and $k_3 = e$ and we obtain
\begin{align}
\begin{aligned}
S_t(A^{k_2},C^{e}) =& -B_t^e \EE_t^{\QQ^e} \left[ \int_{(t,T]} \frac{\mathcal{X}^{e,k_2}_u}{B_u^e} dA_u^{k_2} \right] -B_t^e \EE_t^{\QQ^e} \left[ \int_{(t,T]}\left[ \left( r_u^e - r_u^{c,e,b} \right)(C_u^{k_3})^+  \right. \right.\\
&\left. \left. - \left(r_u^{d_{e}+3,e} - r_u^{c,e,l}  \right)(C_u^{e})^- \right] \frac{1}{B_u^e} du\right].
\end{aligned}
\end{align}
\end{itemize}
\end{remark}

\subsection{Pricing with endogenous collateral}

We treat the case where the collateral depends on the marked-to-market value of the contract. We assume for simplicity that the initial endowment is zero, i.e. $x=0$. We assume again that all interest rates are bounded and we let the filtration $\FF$ be of the form $\FF=\FF^{S,\mathcal{X}}$, i.e. the filtration is generated by all risky assets and exchange rates. In line with the previous section, the contract $A^{e,k_2}$ is adapted to the filtration $\FF^{S,\mathcal{X}}$. The collateral account is now given by
\begin{align}
\label{eq:endoColl}
C^{k_3}_t=(1+\delta^1_t)\frac{\left(-V_t(\varphi)\right)^+}{\mathcal{X}^{e,k_3}_t}-(1+\delta^2_t)\frac{\left(-V_t(\varphi)\right)^-}{\mathcal{X}^{e,k_3}_t},
\end{align}
where the bounded, RCLL $\FF^{S,\mathcal{X}}$-adapted processes $\delta^1,\delta^2$ represent haircuts. The fact that now $C^{k_3}$ depends on $V$ implies that the pricing equation has a recursive nature and hence is to be treated as a BSDE. 

We consider the case of cash collateral with rehypothecation and we further introduce the simplification $r^{d_{k_3}+3,k_3}=r^e$, $r_u^{c,k_3,b}=r_u^{c,k_3,l}=r_u^{c,k_3}$. Concerning the drift of the exchange rate $\mathcal{X}^{e,k_3}$ we can define the cross currency basis $q^{e,k_3}$ via
\begin{align*}
r_t^e - r_t^{k_3}=r_t^{c,e} - r_t^{c,k_3}+q^{e,k_3}_t,
\end{align*}
where $r^{c,e}$ and $r^{c,k_3}$ are the collateral rates under the domestic and the $k_3$ currency. Obviously we have $q^{e,e}\equiv 0$, $q^{e,k_3}=-q^{k_3,e}$. Expressing the dynamics of the exchange rate in terms for the cross currency basis in the present diffuse setting means that we write
\begin{align}
d\mathcal{X}^{e,k_3}_t=\mathcal{X}^{e,k_3}_t\left((r_t^{c,e} - r_t^{c,k_3}+q^{e,k_3}_t)dt+\sigma^{\mathcal{X}^{e,k_3}}_tdW^{\mathcal{X}^{e,k_3},\QQ^e}_t\right).
\end{align}
Under the preceding assumptions \eqref{eq:ccrehypDiff} takes the form
\begin{align}
\label{eq:ccrBSDE}
\begin{aligned}
dV_t(\varphi) &= V_t(\varphi)r_t^e dt + \sum_{k_1=1}^L  \sum_{i=1}^{d_{k_1}} \xi_t^{i,k_1}\mathcal{X}^{e,k_1}_tS^{i,k_1}_t\sigma^{S^{i,k_1}}_tdW^{S^{i,k_1},\QQ^e}_t\\
&+\sum_{k_1=1,k_1 \neq e}^L\psi_t^{0,k_1}B^{k_1}_t\mathcal{X}^{e,k_1}_t\sigma^{\mathcal{X}^{e,k_1}}_tdW^{\mathcal{X}^{e,k_1},\QQ^e}_t  + dA_t^{e,k_2}\\
&+\left( r_t^e - r_t^{c,k_3} \right) \left((1+\delta^1_t)\left(-V_t(\varphi)\right)^+-(1+\delta^2_t)\left(-V_t(\varphi)\right)^-\right)dt\\
&- \left((1+\delta^1_t)\left(-V_t(\varphi)\right)^+-(1+\delta^2_t)\left(-V_t(\varphi)\right)^-\right) \left( r_t^{c,e} - r_t^{c,k_3}+q^{e,k_3}_t \right)dt\\
&- \left((1+\delta^1_t)\left(-V_t(\varphi)\right)^+-(1+\delta^2_t)\left(-V_t(\varphi)\right)^-\right)\sigma^{\mathcal{X}^{e,k_3}}_tdW^{\mathcal{X}^{e,k_3},\QQ^e}_t ,
\end{aligned}
\end{align}
where we substituted also the dynamics the exchange rate expressed via the cross currency basis. We  view the expression above as a BSDE, where the controls are given by the processes
\textcolor{black}{
	\begin{align}
\label{eq:BSDEcontrols}	
\begin{aligned}
Z^{i,k_1}&=\xi^{i,k_1}, \ i=1,\ldots, d_{k_1}, \ k_1=1,\ldots,L\\
Z^{0,k_1}&=\psi_t^{0,k_1}, \ k_1=1,\ldots,L, \ k_1 \neq e\\
Z^{1,k_3}&=- \left((1+\delta^1)\left(-V(\varphi)\right)^+-(1+\delta^2)\left(-V(\varphi)\right)^-\right)
\end{aligned}
\end{align}}
and zero terminal condition. We introduce: 
\begin{itemize}
\item The subspace of all $\RR^d$-valued, $\FF^{S,\mathcal{X}}$-adapted processes $X$ such that  
\begin{align} \label{eq:spaceH2}
\Ex{\QQ^e}{\int_0^T\left\|X_t\right\|^2 dt}<\infty, 
\end{align}
denoted by $\cH^{2,d}(\QQ^e).$ We set $\cH^{2}(\QQ^e) := \cH^{2,1}(\QQ^e).$ 
\item The subspace of all $\RR^d$-valued, $\FF^{S,\mathcal{X}}$-adapted processes $X$ such that 
\begin{align} \label{eq:spaceS2b}
\EE^{\QQ^e}\left[\sup_{t \,\in\, [0,T]} \left\|X_t\right\|^2 \right] < \infty,
\end{align}
denoted by $\cS^{2,d}(\QQ^e).$ We set $\cS^{2}(\QQ^e) := \cS^{2,1}(\QQ^e).$ 
\end{itemize}

We have the following pricing result that follows from \cite{nr2016}. \textcolor{black}{Alternatively, an existence and uniqueness result in a setting where the generator is monotone and the filtration is general can be found in \cite{KruPo2016}. Such results have already been applied in the xVA literature e.g. in \cite{Crepey2016aa}.}

\begin{proposition}\label{propo:ccrPrice}
Assume that $A^{e,k_3}\in\cS^{2}(\QQ^e)$. Then the BSDE \eqref{eq:ccrBSDE} with zero terminal condition admits a unique solution with $V(\varphi)\in\cS^{2}(\QQ^e)$ and the controls \textcolor{black}{\eqref{eq:BSDEcontrols} belong to }$\cH^{2}(\QQ^e)$. Also, the collateralized contract $A^{e,k_3}$ with collateral specification \eqref{eq:endoColl} can be replicated on $[t,T]$ by an admissible trading strategy $\varphi$ and the price admits the representation
\begin{align*}
\begin{aligned}
S_t&(A^{k_2},C^{k_3}) = -B_t^e \EE_t^{\QQ^e} \left[ \int_{(t,T]} \frac{\mathcal{X}^{e,k_2}_u}{B_u^e} dA_u^{k_2} \right]\\
& -B_t^e \EE_t^{\QQ^e} \left[\int_{(t,T]} \left( r_u^e - r_u^{c,e}-q_u^{e,k_3} \right) \frac{\left((1+\delta^1_u)\left(-V_u(\varphi)\right)^+-(1+\delta^2_u)\left(-V_u(\varphi)\right)^-\right)}{B_u^e} du \right].
\end{aligned}
\end{align*}
\end{proposition}

\begin{proof}
Existence and uniqueness to \eqref{eq:ccrBSDE} follow from \cite{nr2016} Theorem 4.1 modulo our assumption that $A^{e,k_3}\in\cS^{2}(\QQ^e)$. The rest of the claim is clear from our previous results.
\end{proof}

It is interesting to study the case of perfect collateralization. This can be immediately obtained from our formulas by setting $\delta^1_t=\delta^2_t=0$ $d\QQ^e\otimes dt$-a.s.. We observe that the BSDE \eqref{eq:ccrBSDE} takes now the much simpler form
\begin{align}
\begin{aligned}
dV_t(\varphi) &= V_t(\varphi)(r_t^{e,c}+q_t^{e,k_3} )dt + \sum_{k_1=1}^L  \sum_{i=1}^{d_{k_1}} \xi_t^{i,k_1}\mathcal{X}^{e,k_1}_tS^{i,k_1}_t\sigma^{S^{i,k_1}}_tdW^{S^{i,k_1},\QQ^e}_t\\
&+\sum_{k_1=1,k_1 \neq e}^L\psi_t^{0,k_1}B^{k_1}_t\mathcal{X}^{e,k_1}_t\sigma^{\mathcal{X}^{e,k_1}}_tdW^{\mathcal{X}^{e,k_1},\QQ^e}_t  + dA_t^{e,k_2}\\
&+V_t(\varphi)\sigma^{\mathcal{X}^{e,k_3}}_tdW^{\mathcal{X}^{e,k_3},\QQ^e}_t,
\end{aligned}
\end{align}
from which, with the help of Proposition \ref{propo:ccrPrice}, we immediately obtain the following valuation formula for perfectly collateralized claims, namely
\begin{align}
\begin{aligned}
S_t(A^{k_2},C^{k_3}) =& -\EE_t^{\QQ^e} \left[ \int_{(t,T]} e^{-\int_t^u r^{c,e}_s+q^{e,k_3}_s ds}\mathcal{X}^{e,k_2}_u dA_u^{k_2} \right],
\end{aligned}
\end{align}
from which we can obtain also the following special cases.
\begin{itemize}
\item[2.a)] \underline{Domestic cashflows collateralized in domestic currency}. This corresponds to the case $k_2 = k_3 = e$ and we obtain
\begin{align}
\begin{aligned}
S_t(A^{e},C^{e}) =& -\EE_t^{\QQ^e} \left[ \int_{(t,T]}e^{-\int_t^u r^{c,e}_s ds} dA_u^{e} \right] ,
\end{aligned}
\end{align}
so we discount using the domestic collateral rate.
\item[2.b)]\underline{Domestic cashflows collateralized in foreign currency}. This corresponds to the case $k_2=e$ and $k_3 \neq e$ and we obtain
\begin{align}
\begin{aligned}
S_t(A^{e},C^{k_3}) =& -\EE_t^{\QQ^e} \left[ \int_{(t,T]} e^{-\int_t^u r^{c,e}_s+q^{e,k_3}_s ds} dA_u^{e} \right] ,
\end{aligned}
\end{align}
so that the foreign collateralization results in the appearance of the cross currency basis in the discount factor.
\item[2.c)]\underline{Foreign cashflows collateralized in domestic currency}. This corresponds to the case $k_2\neq e$ and $k_3 = e$ and we obtain
\begin{align}
\begin{aligned}
S_t(A^{k_2},C^{e}) =& -\EE_t^{\QQ^e} \left[ \int_{(t,T]}e^{-\int_t^u r^{c,e}_s ds} \mathcal{X}^{e,k_2}_udA_u^{k_2} \right].
\end{aligned}
\end{align}
\end{itemize}

\section{Non-linear markets and directions for future research}

\textcolor{black}{The work of \cite{BieRut15} has been continued and further generalized in \cite{biciarut2018}. In this second paper the issue of no arbitrage is further investigated for fully non-linear valuation problems.
We discuss in the following how to relate our results with those of \cite{biciarut2018}. It will become apparent that our results on the multi currency extension of \cite{BieRut15} can be fully integrated in the setting of \cite{biciarut2018}.}

\textcolor{black}{The basic setup to the two papers is very similar. Concerning the concept of bilateral financial contract (our Definition \ref{def:bilContr} or equivalently Definition 2.3 in \cite{BieRut15}) we note that in \cite{biciarut2018}  the initial flow of the contract is not included in their cumulative dividend process $A$ but this is immaterial. Also we note that their definition of dividend process also encodes the set of \textit{trading adjustments}, represented by a family of processes $\mathcal{X}=(X^1,\ldots,X^n)$, so that a bilateral contract is a couple $\mathcal{C}=(A,\mathcal{X})$. This results in a set of modifications of the stream of flows of the contract. This technique has been pioneered by Brigo and co-authors, the most recent example being given by \cite{bbfpr2018}. However, in the context of our paper, trading adjustments can be equivalently formulated in terms of suitable strategies on certain cash accounts as we did in our Section \ref{sec:CollateralizedTrading} without any impact on the results. The equivalence between the two approaches is in fact demonstrated in \cite{bbfpr2018}.}

\textcolor{black}{Definition 1 in \cite{biciarut2018} of a self-financing trading strategy can be directly linked to our Definitions \ref{def:firstSelfFinancing} and \ref{def:secondSelfFinancing}. Also, the definition of the hedger's wealth and that of wealth process of a self-financing trading strategy (Definition 4 in \cite{biciarut2018}) have the same economic meaning (in the single currency framework) of our definitions in Section \ref{sec:CollateralizedTrading}. Regarding the trading activity in risky assets, equation 17 in \cite{biciarut2018} is the single currency analogue of our cross currency repo constraint \eqref{eq:repoConstraint}.}

\textcolor{black}{Concerning the distinction between local and global valuation problems as presented in Definition 7 of \cite{biciarut2018}, a cross currency extension is in general a feature which is independent w.r.t. the local or global nature of the valuation problem, it refers to the currency of denomination of the flows of the contract.  For this reason, it is more convenient to present the multi currency extension of the multiple curve framework in the context of a local valuation problem that leads to classical BSDEs as we do in the present paper. A form of non-linearity in the adjustment process might be introduced via collateral choice options, i.e. by providing the agents with the option to post collateral to the counterparty under any preferred currency. This means in practice that the collateral poster will provide collateral denominated in the currency where he/she has the lowest funding costs. This is however an issue that deserves a separate treatment e.g. by means of stochastic control techniques as in \cite{Piterbarg2013} and we leave it for future research.}

\textcolor{black}{Let us now focus on the concept of absence of arbitrage. Our Definition \ref{def:AoA} represents the multi currency generalization of Definition 3.2 in \cite{BieRut15}. Absence of arbitrage is in this case defined in terms of the process $V^{net}$, i.e. a long-short position where only one transaction is hedged, which is sufficient in the setting we consider. \cite{BieRut15} also hint at the concept of \textit{extended arbitrage opportunity} which is constructed by considering a long-short position of two hedged instruments which is useful when considering e.g. a defaultable setting with two counterparties having a different credit worthiness. This second, more general, arbitrage concept is in fact renamed \textit{arbitrage opportunity for the trading desk} (Definition 13 in \cite{biciarut2018}) whereas the combination of two hedged long/short positions is called \textit{combined wealth} (Definition 11 in \cite{biciarut2018}). The two concepts however share the same economic meaning and, considering the setup that we employ in the present paper, the concept of arbitrage for the hedger is sufficient to construct our cross currency generalization of the multiple curve framework. Also let us notice that the concept of (absence of) arbitrage for the hedger is the one which is (at least implicitly) found in most papers in the literature e.g. in \cite{crepey2015a} and explicitly e.g. in \cite{BiCaStu2018}, so we find it important to present our cross currency generalization via this definition of arbitrage opportunity.  However, generalizing Definition 13, and the no arbitrage criterion in Proposition 2  in \cite{biciarut2018} in line with our work is possible.}

\textcolor{black}{In summary, the results of the present paper can be generalized to cover nonlinear market models as treated in \cite{biciarut2018}, however the contributin of the present paper is sufficient to provide a sound foundation for a theory of cross currency markets in the context of the multiple curve framework. We conjecture that some advanced features of collateral agreements, e.c. collateral choice options, could lead to global valuation problems in the sense of Definition 7 in \cite{biciarut2018} with associated generalized BSDEs as proposed in \cite{cheridito2017}. We leave such investigations for future research.}

 \bibliographystyle{apa}
\bibliography{references}
 
\end{document}